\documentclass[twocolumn,english,journal]{IEEEtran}
\ifCLASSINFOpdf
\else
\fi

\usepackage[T1]{fontenc}
\usepackage{babel}
\usepackage{array}
\usepackage{calc}
\usepackage{amsmath}
\usepackage{amsthm}
\usepackage{amssymb}
\usepackage{graphicx}
\usepackage{cite}
\usepackage{algorithm}
\usepackage{subfigure}
\usepackage{color}

\usepackage{enumitem}
\usepackage{multirow}
\usepackage{hhline}

\newtheorem{theorem}{Theorem}

\usepackage[unicode=true,
bookmarks=true,bookmarksnumbered=true,bookmarksopen=true,bookmarksopenlevel=1,
breaklinks=false,pdfborder={0 0 0},backref=false,colorlinks=false]
{hyperref}
\hypersetup{pdftitle={Your Title},
	pdfauthor={Your Name},
	pdfpagelayout=OneColumn, pdfnewwindow=true, pdfstartview=XYZ, plainpages=false}

\makeatletter

\providecommand{\tabularnewline}{\\}

\let\oldforeign@language\foreign@language
\DeclareRobustCommand{\foreign@language}[1]{%
	\lowercase{\oldforeign@language{#1}}}

\makeatother

\allowdisplaybreaks

\makeatletter
\newcommand{\thickhline}{%
	\noalign {\ifnum 0=`}\fi \hrule height 1.5pt
	\futurelet \reserved@a \@xhline
}
\newcolumntype{"}{@{\hskip\tabcolsep\vrule width 1pt\hskip\tabcolsep}}
\makeatother

\begin{document}
%
\title{Can Attackers with Limited Information Exploit Historical Data to Mount Successful False Data Injection Attacks on Power Systems?}
%
%
%

\author{Jiazi~Zhang,~\IEEEmembership{Student~Member,~IEEE,}
        Zhigang~Chu,~\IEEEmembership{Student~Member,~IEEE,}
        Lalitha~Sankar,~\IEEEmembership{Senior~Member,~IEEE,}
        and~Oliver~Kosut,~\IEEEmembership{Member,~IEEE}
\thanks{The authors are with the School of Electrical, Computer, and Energy
	Engineering, Arizona State University, Tempe, AZ 85281 USA (e-mail:
	\{jzhan188,zchu2,lalithasankar,okosut\}@asu.edu).}}

%
%

\markboth{IEEE TRANSACTIONS ON POWER SYSTEMS}%
{Zhang \MakeLowercase{\textit{et al.}}: False Data Injection Attacks with Limited Information}
%



\maketitle

\begin{abstract}
	This paper studies physical consequences 
	of unobservable false data injection (FDI) attacks designed only with information inside a sub-network of the power system. The goal of this attack is to overload a chosen target line without being detected via measurements. To overcome the limited information, a multiple linear regression model is developed to learn the relationship between the external network and the attack sub-network from historical data.
	The worst possible consequences of such FDI attacks are evaluated by solving a bi-level optimization problem wherein the first level models the limited attack resources, while the second level formulates the system response to such attacks via DC optimal power flow (OPF). The attack model with limited information is reflected in the DC OPF formulation that only takes into account the system information for the attack sub-network. The vulnerability of this attack model is illustrated on the IEEE 24-bus RTS and IEEE 118-bus systems.
\end{abstract}


%
\IEEEpeerreviewmaketitle
\global\long\def\figurename{Fig.}
\global\long\def\tablename{Table}

%
%
%
%
\section*{Nomenclature}
\hspace{-0.4cm}\textbf{Topologies}

	\begin{description}[leftmargin=1.8cm,style=multiline]
		\item[$\mathcal{E}$] the area outside $\mathcal{L}$, where the attacker has no knowledge, \textit{i.e.,} $\mathcal{E}=\mathcal{G}\setminus \mathcal{L}$; 
		\item[$\mathcal{G}$] the entire network of the test system;
		\item[$\mathcal{L}$] a sub-network of $\mathcal{G}$ bounded by load buses where the attacker has perfect knowledge inside it;
		\item[$\mathcal{S}$] a sub-graph of $\mathcal{G}$ bounded by load buses where the attacker may replace measurements inside it.
	\end{description}

	\hspace{-0.4cm} The characters $\mathcal{G}$, $\mathcal{S}$, $\mathcal{L}$, and $\mathcal{E}$ also represent the sets of buses inside the corresponding networks.
	
	\hspace{-0.4cm}\textbf{Sets}
		\begin{description}[leftmargin=1.8cm,style=multiline]
	\item[$\mathcal{B}$] the set of boundary buses in $\mathcal{L}$ such that each bus in $\mathcal{B}$ is connected to at least one bus in $\mathcal{E}$;
	\item[$\mathbb{C}$] the set of lines in $\mathcal{G}$ that are congested for each instance of the historical data, where $\mathbb{C}^+$ and $\mathbb{C}^-$ are the subsets in $\mathbb{C}$ for which the power flow directions are positive and negative, respectively;
	\item[$\mathcal{E}_M$] the set of buses with marginal generators in $\mathcal{E}$;
	\item[$\mathbb{G}_i$] the set of generators that are connected to bus $i$;
	\item[$\mathcal{I}$] the set of internal buses, \textit{i.e.,} $\mathcal{I}=\mathcal{L} \setminus \mathcal{B}$;	
	\item[$\mathbb{W}_i$] the set of lines that are connected to boundary bus $i$, where $\mathbb{W}_i^{\mathcal{L}}$ and $\mathbb{W}_i^{\mathcal{E}}$ are the subsets of lines located in $\mathcal{L}$ and $\mathcal{E}$, respectively;	
	\item[$\mathcal{Y}$] the subset of remaining buses in $\mathcal{E}$, \textit{i.e.,} $\mathcal{Y}=\mathcal{E}\setminus \mathcal{Z}$;
	\item[$\mathcal{Z}$] the subset of buses in $\mathcal{E}$, for which power injections remain constant in the historical data.
\end{description}

\hspace{-0.4cm}\textbf{Quantities \& Indices}
\begin{description}[leftmargin=1.8cm,style=multiline]
	\item[$l$] target line index;
	\item[$n_b$] number of buses in $\mathcal{G}$;
	\item[$n_{br}$] number of lines in $\mathcal{G}$;
	\item[$n_c$] number of lines in $\mathbb{C}$;
	\item[$n_g$] number of generators in $\mathcal{G}$;
	\item[$n_z$] number of measurements in $\mathcal{G}$;
	\item[$n_{\mathcal{B}}, n_{\mathcal{L}}, n_{\mathcal{E}},$ $n_{\mathcal{Y}}$] number of buses in sets $\mathcal{B}$, $\mathcal{L}$, $\mathcal{E}$, and $\mathcal{Y}$, respectively.
	\end{description}
	
	\hspace{-0.4cm}\textbf{Parameters}
		\begin{description}[leftmargin=1.8cm,style=multiline]
	\item[$A, B$] the $(n_c+1)\times n_{\mathcal{L}}$ and $(n_c+1)\times n_{\mathcal{Y}}$ coefficient matrices associated with $v_{\mathcal{L}}$ and $v_{\mathcal{Y}}$ in \eqref{eq:ABF}, respectively, \textit{i.e.,} $A=[K_{\mathbb{C}}^\mathcal{L}; \mathbf{1}^T]$ and $B=[K_{\mathbb{C}}^{\mathcal{Y}}; \mathbf{1}^T]$;
	\item[$C_g(\cdot)$] the quadratic cost function for generator $g$;
	\item[$d$] the constant vector in \eqref{eq:ABF} such that $d=[SP_{\mathbb{C},\text{max}}-K_{\mathbb{C}}^{\mathcal{Z}}v_{\mathcal{Z}}; -\mathbf{1}^T v_{\mathcal{Z}}]$; 
	\item[$G$] the $n_b \times n_g$ generator-to-bus connectivity matrix in $\mathcal{G}$; 
	\item[$H$] the $n_b \times n_b$ dependency matrix between power injection measurements and state variables in $\mathcal{G}$;
	\item[$I$] identical matrix;
	\item[$J_i$] the $1\times n_b$ row vector of the sum of row vectors in $K$ corresponding to lines in $\mathbb{W}_i^{\mathcal{E}}$, \textit{i.e.,} $J_i=\sum\limits_{k\in\mathbb{W}_i^{\mathcal{E}}}K_k$;
	\item[$K$] the $n_{br}\times n_b$ power transfer distribution factor (PTDF) matrix in $\mathcal{G}$;
	\item[$M$] a large constant;
	\item[$N_0$] the $l_0$-norm constraint limit;
	\item[$N_1$] the $l_1$-norm constraint limit;
	\item[$P_{\text{max}}$] the $n_{br}\times 1$ thermal limit vector in $\mathcal{G}$;
	\item[$P_D$] the $n_b\times 1$ real power load vector in $\mathcal{G}$;
	\item[$P_{G,\text{max}},$ $P_{G,\text{min}}$] the $n_g \times 1$ maximum and minimum generation limit vectors, respectively, in $\mathcal{G}$;
	\item[$S$] a $n_c \times n_c$ diagonal matrix with $S_{kk}=1$, $\forall k\in \mathbb{C}^+$, and $S_{kk}=-1$, $\forall k\in\mathbb{C}^-$;
	\item[$\Gamma$] the $n_{br}\times n_b$ dependency matrix between power flow measurements and voltage angle states in $\mathcal{G}$;
	\item[$\lambda_i$] the locational marginal price (LMP) at boundary bus $i$;	
	\item[$\tau$] the load shift factor;
	\item[$\zeta$] the weight of the norm of attack vector $c$.
\end{description}
	
\hspace{-0.4cm}\textbf{Variables}
	\begin{description}[leftmargin=1.8cm,style=multiline]
	\item[$c$] the $n_{b}\times1$ attack vector in $\mathcal{G}$;
	\item[$e$] the $n_z \times 1$ vector of measurement noise in $\mathcal{G}$;
	\item[$P$] the $n_{br}\times1$ vector of branch power flow in $\mathcal{G}$;
	\item[$P_l^{\text{ub}}$] the upper bound of the physical power flow on target line $l$ resulting from attack vector $\bar{c}^*$;
	\item[$P_{G}$] the $n_{g}\times1$ vector of generation dispatch variable in $\mathcal{G}$;
	\item[$\bar{P}_{I,i}, \bar{P}_{I,\mathcal{B}}$] the pseudo-boundary injection variable at boundary bus $i$, and the $n_\mathcal{B}\times 1$ pseudo-boundary injection vector of all boundary buses, respectively;
	\item[$u$] the $n_b\times 1$ slack vector to linearize the $l_1$-norm constraint;
	\item[$v$] the $n_b \times 1$ vector of power injections in $\mathcal{G}$;	
	\item[$x$] the $n_b \times 1$ vector of state variables in $\mathcal{G}$;
	\item[$z$] the $n_z \times 1$ vector of measurements in $\mathcal{G}$;
	\item[$\delta$] the vector of binary variables for dual variables in the bi-level attack optimization problem.
\end{description}

	\hspace{-0.4cm}\textbf{Multiple Linear Regression Parameters}
	\begin{description}[leftmargin=1.8cm,style=multiline]
	\item[$\mathbf{x}_t$] the input vector at time instance $t$, \textit{i.e.,} $\mathbf{x}_t^T= \bar{G}\bar{P}_{G,t} - \bar{P}_{D,t}$;
	\item[$\mathbf{X}$] the input matrix including input vectors for $t=1,2,\hdots,m$;
	\item[$\mathbf{y}_i$] the observed output vector at boundary bus $i$ for $t=1,2,\hdots,m$, \textit{i.e.,}, $y_{i,t}=\bar{P}_{I,i,t}$;
	\item[$\varepsilon_{i,t}$] the random error at boundary bus $i$ at time instance $t$;
	\item[$\hat{F}_i, \hat{f}_{i,0}$] $[\hat{F}_i, \hat{f}_{i,0}]$ is the estimated regression coefficient vector with input matrix $\mathbf{X}$ and output vector $\mathbf{y}_i$, \textit{i.e.,} $[\hat{F}_i, \hat{f}_{i,0}]^T=\left(\mathbf{X}^T\mathbf{X}\right)^{-1}\mathbf{X}^T\mathbf{y}_i$. 
	\end{description}
	
\hspace{-0.35cm}For any vector or matrix associated with the entire network $\mathcal{G}$ such as $c, G, H, K, P, P_G, P_D$, and $\Gamma$, we write the equivalent parameters corresponding to the sub-network $\mathcal{L}$ with $\bar{c}, \bar{G}, \bar{H}, \bar{K}, \bar{P}, \bar{P_G}, \bar{P}_D$, and $\bar{\Gamma}$, respectively.

\section{Introduction}

\IEEEPARstart{T}{he} electric power system is monitored via an extensive network of sensors in tandem with data processing algorithms, \textit{i.e.}, an intelligent cyber layer, that enables continual observation and control of the physical system. In recent years, several incidents \cite{Zeller2011,StuxnetIran,Siemens14plants,UkraineAttack} demonstrate that the cyber layer of power system is vulnerable to cyber-attacks that impact the system operation status and lead to serious physical consequences. Therefore, it is crucial to fully understand the potential consequences of such attacks.

There is an overall need to understand the effects of cyber-attacks against power systems. Doing so requires considering different categories of attacks, and which of them are credible and could lead to severe consequences. Fig. 1 illustrates our view of the space of cyber-attacks.  Within the space of all cyber-attacks, we choose to focus on false data injection (FDI) attacks in which the attacker replaces a subset of measurements with counterfeits.
The most effective FDI attacks are those which are unobservable to state estimation (SE). An unobservable FDI attack is one for which the measurements look like they originated from legitimate state values but are in fact spoofed by the attacker. Thus, such an attack cannot be detected as noisy measurements by bad data detector. It has been established in \cite{Liu2009,Hug2012,Liang2014} that FDI attacks can bypass both DC SE and AC SE when designed appropriately. Such attacks can be designed to specifically have an impact on electricity markets (\textit{e.g.}, \cite{Yuan11}, \cite{Yuan2012}) via a bi-level optimization problem wherein the first level problem models the attacker's goal and the second level problem models the system response. 

However, this did not address whether credible FDI attacks can be constructed that would actually lead to severe consequences. Our goal is to understand this question. The earlier work of \cite{Liang2015} and \cite{Zhang2016TSG} focused on the question of severe consequences: can an unobservable FDI attack cause not only a spoofed state estimator but also physical consequences? The answer was yes, in principle, by causing a control center to incorrectly dispatch, which would in turn cause a line overflow. Such attacks, similar to those in \cite{Yuan11} and \cite{Yuan2012}, can also be designed via a bi-level optimization problem.
This attack optimization problem requires the attacker to know system-wide information including topology, generation cost and capacity, and load data. In practice, obtaining all the required information can be difficult for the attacker. In order to ensure that the worst-case attacks in \cite{Yuan11,Yuan2012,Liang2015,Zhang2016TSG} are credible, we focus on understanding whether it is possible at all to design FDI attacks with only limited system information in this paper. 
Recently, \cite{Rahman2012,Liu2014TSG,Liu2016TSG} have demonstrated that it is possible to design FDI attacks against SE with inaccurate or limited topology information. However, physical consequences of the worst-case limited information FDI attacks have not been analyzed. 
	\begin{figure}[h]
		\centering{}\includegraphics[scale=0.7]{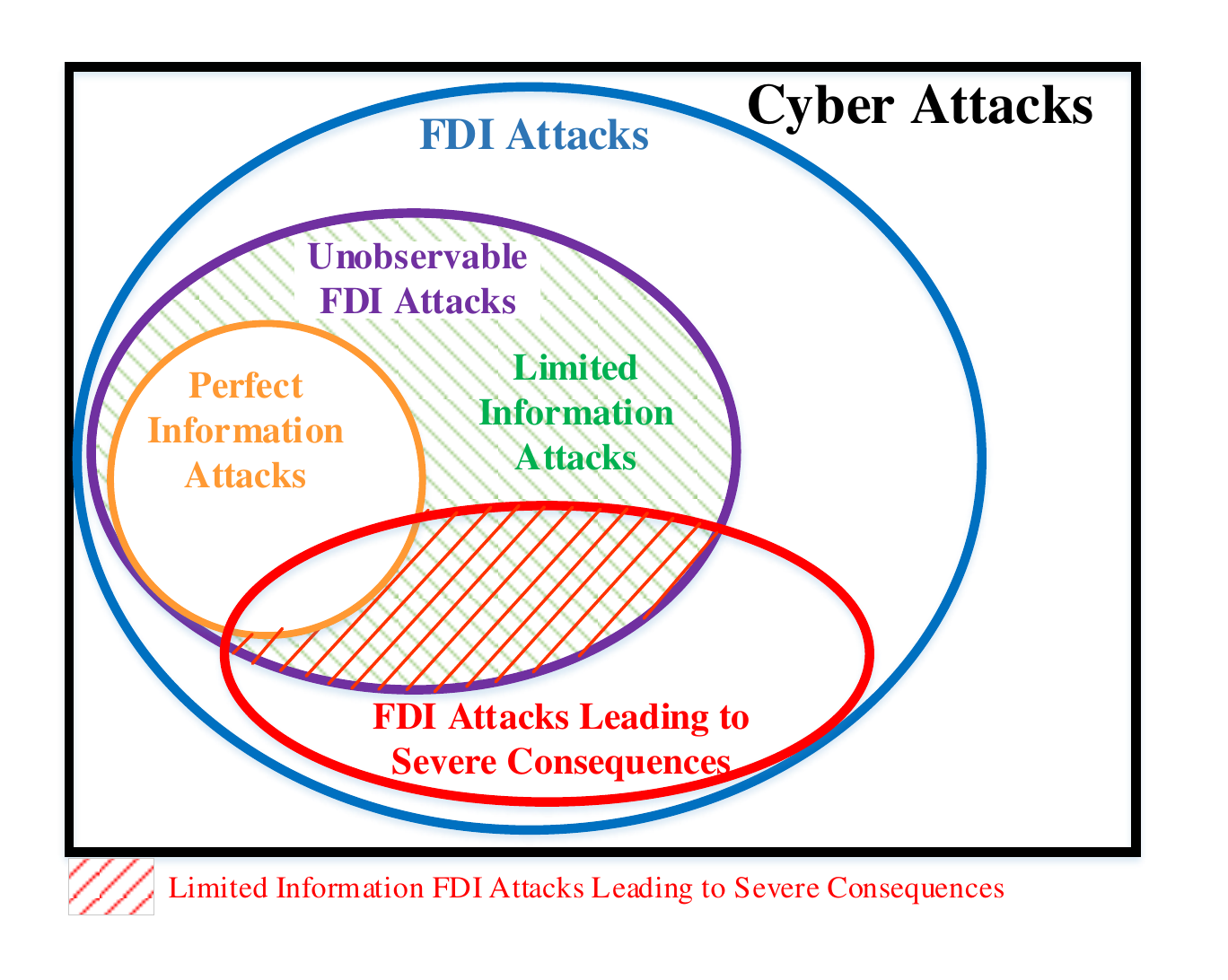}\protect\protect\caption{The space of cyber-attacks. \label{fig:AttacksPlot}}
	\end{figure}

In this paper, we assume the attacker only has access to information inside an attack sub-network and absolutely no knowledge of the outside network. In order to overcome the limited information, we suppose that the attacker infiltrates the sub-network long before it executes its attack, so that it can observe the natural behavior of the system in order to predict the effect of an attack. In particular, we assume that the attacker has access to historical data inside the sub-network that includes loads, costs, capacities, status, and dispatches of generators, and locational marginal prices. Historical data is sometimes directly utilized as pseudo-measurements to SE when the real-time information is incomplete. However, the attacker can be more sophisticated and use historical data to create higher fidelity boundary pseudo-measurements when they only have limited information. In this work, we suppose that the attacker uses multiple linear regression method to learn the relationship between the external network and the attack sub-network from historical data. Furthermore, we predict the response of the control center under such attacks in a local sub-network via a bi-level optimization problem. 

The limited information attack problem is similar to the "seamless" market problem  \cite{SeamsPESERC2010} which aims to achieve maximum social welfare across several adjacent markets, while allowing each market to model its own system, exchanging boundary information with its neighbors. In order to predict the behavior of the adjacent market, in \cite{Farrokhseresht16,SeamsPESERC2010}, linear regression model is used. However, in contrast to the market problem that requires perfect prediction of the external network, the attacker only needs partial prediction to overload a target line. In this paper, we demonstrate that even if the prediction of the external network re-dispatch is inaccurate, the attacker can still cause overflow on the target line. 

This paper builds upon our prior work \cite{ZhangPES2016}, where we consider limited information such that the attacker can learn system parameters perfectly inside a sub-network and imperfectly outside of the sub-network. In \cite{ZhangPES2016}, we demonstrate that such attacks can result in line overflows with both accurate and inaccurate information outside of the attack sub-network. 


The contributions of this paper are as follows:
\begin{enumerate}
	\item We introduce a method to compute power flows inside the attack sub-network with only localized information by approximating the external effects via \textit{pseudo-boundary injections}.
	We develop a multiple linear regression model, allowing the attacker to learn the relationship between pseudo-boundary injections and power injections in the sub-network from historical data.
	\item We introduce a bi-level optimization problem from the attacker's perspective to maximize the power flow on a target line wherein the first level models attacker's limited resources and the second level models system response via a modified DC OPF. Such DC OPF formulation only takes into account localized information and the pseudo-boundary injections. From the perspective of the system, this bi-level optimization problem can also be employed as the vulnerability analysis to evaluate the sub-graph that are prone to be attacked.
	\item We take the equivalent constraints that satisfy all historical data into account to prove the existence of a linear relationship between pseudo-boundary injections and power injections inside the attack sub-network for certain topologies. Note that the linear relationship here is different from the Ward \cite{Ward1949} or Monticelli \cite{Monticelli1979} network equivalent which is derived from a specific power flow case. 
	Furthermore, we show that even if the attacker cannot exactly predict the physical consequences, it can compute an upper bound on the power flow of the target line.	
	\item We demonstrate that an attacker can cause line overflows in the IEEE 24-bus RTS system and IEEE 118-bus system using this bi-level attack optimization problem. 
\end{enumerate}  


The paper is organized as follows. The system and attack models are introduced in Sec.~II. Prior work on perfect information FDI attack is reviewed in Sec.~III. The limited information attack model is presented in Sec.~IV. Justifications for the proposed attack strategy is given in Sec.~V, followed by numerical results in Sec.~VI and conclusions in Sec.~VII.



\section{System and Attack Models}
In this section, we introduce the mathematical formulation
for SE, unobservable FDI attacks, and
OPF. Throughout, we assume there are $n_b$ buses,
$n_{br}$ lines, $n_g$ generators, and $n_z$ measurements in the
system. We assume that the system uses DC SE, and DC OPF. The topology of the entire power system is denoted by $\mathcal{G}$.

\subsection{Measurement Model and State Estimation}
	The DC measurement model can be written as
	\begin{equation}
	z=Hx+e\label{eq:DCMeasurement}
	\end{equation}
	where $z$ is the $n_{z}\times1$ measurement vector; $x$ is the $n_b\times1$ voltage angle state vector; $H$ denotes the $n_{z}\times n_{b}$ dependency matrix between measurements and states; $e$ is the $n_{z}\times1$ measurement error vector assumed to be composed of independent Gaussian random
	variables. 
	
	We use weighted least-squares (WLS) to solve this problem \cite{AburBook}. Subsequent  to  SE, a bad data detector uses $\chi^2$-test to detect and eliminate noisy measurements. 
	
\subsection{Unobservable FDI Attack Model}
	In an unobservable FDI attack, the attacker aims to maliciously change the system states from $x$ to $x+c$ without being being detected by bad data detector. In the absence of noise, the measurements after such attacks, $z^a$, satisfy
		\begin{equation}
		z^a=z+Hc=H\left(x+c\right)
		\end{equation}
		where $c$ is the $n_b\times 1$ attack vector. 
\subsection{\label{sub:E_opf}Optimal Power Flow}

The DC OPF problem aims to minimize the total costs of all generators subject to power balance, thermal limits, and generation limits constraints. The formulation of OPF is described in detail in Sec. \ref{sub:GlobalAttack}. 


\section{\label{sec:AttackModel}Prior Work: Perfect Knowledge FDI Attacks}
In this section, we briefly review a closely related work \cite{Liang2015} on unobservable FDI attacks assuming an attacker with perfect knowledge.  As in \cite{Liang2015}, we distinguish between two types of buses in the network: \textit{load buses} that have load directly connected to that bus, and \textit{non-load buses} with no load.
The knowledge (denoted K1) and capabilities (denoted C1) of the attacker in \cite{Liang2015} is described below:

\begin{description}\label{des:AttackCapacity}
	\item[K1.] The attacker has knowledge of (i) the complete network topology; (ii) the cost, capacity, and operational status of all generators in the system; and (iii) historical load data of the entire network.
	\item[C1.] The attacker may choose a small area $\mathcal{S}$, which is a \textit{sub-graph} of the entire network $\mathcal{G}$, \emph{i.e.}, a sub-network chosen in certain manner (see below for description) and bounded by load buses. The attacker may replace measurements inside $\mathcal{S}$.
\end{description}


\subsection{\label{sub:GlobalAttack}Attack Design with Perfect Information}
In \cite{Liang2015}, a bi-level optimization problem is introduced to find the FDI attack that maximizes the power flow on a chosen target line. In \cite{Liang2015}, the authors use $B$-$\theta$ method to formulate the second level DC OPF, in which the power flow vector is computed as a linear function of voltage angle, $\theta$. In contrast, in this paper we equivalently formulate the DC OPF using PTDF, where the line power flow is calculated as the product of PTDF matrix and power injection. The bi-level attack optimization problem is as follows: 
\begin{flalign}
\underset{c,P}{\text{maximize}}\;\; \hspace{0.17cm} & P_{l}-\zeta\left\Vert c\right\Vert _{0}\label{eq:Obj1_MaxPF}\\
\notag \text{subject to}\hspace{0.2cm}\;\\
& \hspace{-0.9cm}P=K(GP_{G}^{*}-P_{D}) \label{eq:Physical_PF}\\
& \hspace{-0.9cm}\left\Vert c\right\Vert _{0}\leq N_{0}\label{eq:con_resources}\\
& \hspace{-0.9cm}-\tau P_{D}\leq Hc\leq \tau P_{D}\label{eq:con_loadshift}\\
& \hspace{-0.9cm}\left\{P_{G}^{*}\right\} =\text{arg}\left\{\underset{P_G}{\text{minimize}}\;\; \underset{g=1}{\overset{n_{g}}{\sum}} C_{g}\left(P_{Gg}\right)\right\} \label{eq:OBJ_MINCOST}\\
&\notag \hspace{-0.8cm} \text{subject to}\\
&\hspace{-0.2cm}
\sum_{g=1}^{n_{g}}P_{Gg}=\sum_{i=1}^{n_{b}}P_{Di} \label{eq:con_nodebalance}\\
&\hspace{-0.5cm}-P_\text{max}\leq K(GP_{G}-P_{D}+Hc)\leq P_\text{max}  \label{eq:con_powerflow}\\
&\hspace{-0.2cm} P_{G,\text{min}}\leq P_{G}\leq P_{G,\text{max}} \label{eq:con_GENlimit}
\end{flalign}
where $P$ is a $n_{br}\times 1$ vector of power flow with thermal limit as  $P_{\text{max}}$, $P_G$ is the $n_g\times 1$ active power generation vector with maximum and minimum limits as  $P_{G,\text{max}}$ and $P_{G,\text{min}}$, respectively; $G$ is the $n_b \times n_g$ generator-to-bus connectivity matrix; $C_g(\cdot)$ is the quadratic cost function for generator $g$; $K$ is the $n_{br}\times n_{b}$ power transfer distribution factor (PTDF) matrix;  $H$ is the $n_{b}\times n_{b}$ dependency matrix between power injection measurements and state variables; $P_{D}$ is the $n_{b}\times 1$ real power load vector;  $\tau$ is the load shift factor which represents the percentage that the cyber load (computed with the spoofed measurements) differs from the physical load; $N_{0}$ is the $l_{0}$-norm constraint limit; and $\zeta$ is the weight of the norm of attack vector $c$.

The objective of the optimal attack problem is to maximize the power flow on the target line $l$ while changing as few states as possible. In the first level, the attack vector is chosen subject to the $l_{0}$-norm constraint of the attack vector in \eqref{eq:con_resources}, \textit{i.e.,} the number of non-zero elements in the attack vector $c$, and the load shift limitation in \eqref{eq:con_loadshift}.  In the second level, the system response to the attack determined in the first level is modeled via DC OPF in \eqref{eq:OBJ_MINCOST}$-$\eqref{eq:con_GENlimit}, where the objective in \eqref{eq:OBJ_MINCOST} is to minimize the the total costs of all generators and constraints \eqref{eq:con_nodebalance}$-$\eqref{eq:con_GENlimit} represent power balance, thermal limits, and generation limits. Note that the power injection vector in \eqref{eq:con_powerflow} is changed from $(GP_G-P_D)$ to $(GP_G-P_D+Hc)$ by the attacker. 

The bi-level optimization problem introduced above is non-linear and
non-convex. For tractability, several constraints are modified to convert the original
formulation into an equivalent mixed-integer linear program (MILP). 
The modifications include: (a) relaxing the $l_{0}$-norm constraint
in \eqref{eq:con_resources} to an $l_{1}$-norm constraint with limit $N_{1}$ and linearizing it by introducing a slack vector $u$; (b) replacing
the second level DC OPF problem by its Karush-Kuhn-Tucker
(KKT) optimality conditions; and (c) linearizing the complementary slackness conditions in KKT by introducing a new vector $\delta$ of binary variables for dual variables and a large constant $M$.


\subsection{Attack Implementation}
Once the attack vector $c$ is determined, the attacker can identify the buses with state changes to enable the attack, \textit{i.e.,} the buses corresponding to non-zero entries of $c$. We refer to these buses as \textit{center buses}. The attack sub-graph $\mathcal{S}$ includes all center buses as well as the lines and buses connecting to every center bus. The non-center buses connected to a center bus are all load buses. 
This method ensures that nothing is changed outside the attack sub-graph $\mathcal{S}$ while the changes needed at the non-center buses are presented as load changes.

Given attacker's knowledge K1 and capabilities C1, the authors in \cite{Liang2015} introduce the FDI attacks as follows:
\begin{align} 
z_i^a & =\begin{cases}
\begin{array}{l}
z_{i}\:,\\
z_{i} + H_i c\:,
\end{array} & \begin{array}{l}
i\notin\mathcal{S}\\
i\in\mathcal{S}
\end{array}\end{cases}.\label{eq: ModMeasureStra1} 
\end{align}

Note that, this attack may not be unobservable to AC SE \cite{Jia2012}, but can be converted to an unobservable AC attacks as 
\begin{align} 
z_i^a =\begin{cases}
\begin{array}{l}
z_{i}\:,\\
h_i(\hat{x}+ c)\:,
\end{array} & \begin{array}{l}
i\notin\mathcal{S}\\
i\in\mathcal{S}
\end{array}\end{cases}\label{eq:DCtoAC}
\end{align}	
 where $h_i(\cdot)$ is the non-linear relationship between measurement $z_i$ and state vector $x$, $\hat{x}$ is the state vector that the attacker estimated with measurements in $\mathcal{S}$.
 The method is first introduced in \cite{Hug2012} and \cite{Liang2015}. Furthermore, in our prior work \cite{Liang2015, Zhang2016TSG, ZhangPES2016}, we have demonstrated that the consequences of the AC attacks track those of the original DC attacks.



\section{\label{sec:LocalAttack}Optimal Line Overflow Attacks with Localized Information}
In this section, we build upon the attack in \cite{Liang2015} by replacing assumptions K1 and C1 with the following limited assumptions on the attacker's knowledge (K2) and capability (C2):
\begin{description}\label{des:AttackCapacity1}
	\item[K2] Within a sub-network $\mathcal{L}$, the attacker has perfect knowledge of the topology, historical load data, generator data including operational status, capacity, cost, and historical dispatch information, and locational marginal price (LMP). In particular, we assume that the attacker has enough historical data to perform the multiple linear regression described in the sequel. This sub-network $\mathcal{L}$ is bounded by load buses.
	\item[C2] The attacker may modify measurements within an attack sub-graph $\mathcal{S}$ within $\mathcal{L}$, \textit{i.e.,} $\mathcal{S}\subseteq\mathcal{L}$. 
\end{description}
An example attack sub-network in the IEEE 24-bus RTS system is shown in Fig. \ref{fig:LocalAttack}.

\textit{Notation}: Recall that the area outside $\mathcal{L}$, where the attacker has no knowledge, is denoted as the external network, \textit{i.e.}, $\mathcal{E}=\mathcal{G} \setminus \mathcal{L}$.  We define the set of boundary buses in $\mathcal{L}$ as $\mathcal{B}$, such that each bus in $\mathcal{B}$ is connected to at least one bus in $\mathcal{E}$. The set of remaining buses in $\mathcal{L}$ is defined as the internal bus set $\mathcal{I}=\mathcal{L} \setminus \mathcal{B}$. For any vector or matrix associated with the entire network $\mathcal{G}$ such as $c,G,H,K,P,P_G$, and $P_D$, we write the equivalent parameters corresponding to the sub-network $\mathcal{L}$ with $\bar{(\cdot)}$. For example, $\bar{H}$ refers to the dependency matrix between power injection measurements and state variables only inside $\mathcal{L}$. Sub-vectors are denoted by subscripts with the corresponding set of elements (buses, lines, or generators). Sub-matrices are denoted with a subscript giving the set of rows, and a superscript giving the set of columns.


\subsection{\label{sub:PF_LOCAL}System Power Flow with Localized Information}
According to assumption K2, the attacker only has knowledge inside $\mathcal{L}$. Therefore, the attacker cannot calculate the line power flow inside $\mathcal{L}$ with \eqref{eq:Physical_PF} since both the PTDF matrix $K$ of the network $\mathcal{G}$ and the subset of power injections in external network $\mathcal{E}$ are unavailable to attacker. To form the line power flow with K2, we introduce a vector of \textit{pseudo-boundary injection} $\bar{P}_{I,\mathcal{B}}$. The $i$th entry of $\bar{P}_{I,\mathcal{B}}$, namely $\bar{P}_{I,i}$, corresponding to boundary bus $i$, represents the sum of power flows delivered from $\mathcal{L}$ to $\mathcal{E}$ at boundary bus $i$, $i\in \mathcal{B}$, as 
\begin{equation}
\bar{P}_{I,i} = \sum_{k\in\mathbb{W}^{\mathcal{E}}_i}P_k \label{eq:PseudoInjection_PKE}
\end{equation}
where $\mathbb{W}^{\mathcal{E}}_i$ represents the lines located in $\mathcal{E}$ that are connected to boundary bus $i$. 

Using \eqref{eq:PseudoInjection_PKE}, the vector of line power flows in $\mathcal{L}$ can be written as
\begin{equation}
\bar{P} = \bar{K}^{\mathcal{I}}(\bar{G}_{\mathcal{I}}\bar{P}_G-\bar{P}_{D,\mathcal{I}})+\bar{K}^{\mathcal{B}}(\bar{G}_{\mathcal{B}}\bar{P}_G-\bar{P}_{D,\mathcal{B}}-\bar{P}_{I,\mathcal{B}}) \label{eq:LocalPF}
\end{equation}
where $\bar{K}$ is split into column-wise sub-matrices $\bar{K}^{\mathcal{I}}$ and $\bar{K}^{\mathcal{B}}$, and $\bar{G}$ is split into row-wise sub-matrices $\bar{G}_{\mathcal{I}}$ and $\bar{G}_{\mathcal{B}}$, both corresponding to buses in $\mathcal{I}$ and $\mathcal{B}$, respectively.
This equation can be further simplified as
\begin{equation}
\bar{P} = \bar{K}(\bar{G}\bar{P}_G-\bar{P}_D)-\bar{K}^{\mathcal{B}}\bar{P}_{I,\mathcal{B}}. \label{eq:LocalPF_b}
\end{equation}



\subsection{\label{sub:MLR}Multiple Linear Regression}
The optimal line overflow attack introduced in Sec. \ref{sub:GlobalAttack} involves determining the attack vector in the first level and estimating the system response to the attack via the whole system DC OPF in the second level. However, due to limited knowledge, the attacker must predict the response of the OPF using only local knowledge. The OPF may be reformulated to include power balance, thermal limit, and generation limit constraints only in $\mathcal{L}$, and apply \eqref{eq:LocalPF_b} to capture all effects in the external network through the pseudo-boundary injections $\bar{P}_{I,\mathcal{B}}$. However, with this formulation, the attacker still cannot predict how the attack affects $\bar{P}_{I,\mathcal{B}}$ since it depends on both power injections in $\mathcal{L}$ and $\mathcal{E}$. Therefore, before the attack is executed, the attacker cannot estimate the system re-dispatch after the attack accurately.       

If the attacker can obtain a large amount of historical power injections and pseudo-boundary injections data in $\mathcal{L}$ (for example, by observing the system over a long time), it can learn a functional relationship between pseudo-boundary injection, $\bar{P}_{I,\mathcal{B}}$, and power injections inside $\mathcal{L}$. The attacker can then predict the pseudo-boundary injections with the power injection in $\mathcal{L}$ as 
\begin{equation}
\hat{\bar{P}}_{I,\mathcal{B}}=\hat{F}\left(\bar{G}\bar{P}_G-\bar{P}_D\right)+\hat{f}_0 \label{eq:MLP_coeff}
\end{equation}   
where $[\hat{f}_0 \; \hat{F}]$ represent an affine relationship, and $\hat{\bar{P}}_{I,\mathcal{B}}$ is the attacker's prediction of pseudo-boundary injection by capturing the functional relationship via a linear model. Note that the historical pseudo-boundary injections can be computed with data in $\mathcal{L}$ as
\begin{equation}
\bar{P}_{I,i}=\sum\limits_{g\in \mathbb{G}_i}\bar{P}_{G,g}-\bar{P}_{D,i}-\sum\limits_{k\in \mathbb{W}^{\mathcal{L}}_i} \bar{P}_k \label{eq:PseudoInjection}
\end{equation}  
where $\mathbb{G}_i$ is the set of generators connected to bus $i$, and $\mathbb{W}^{\mathcal{L}}_i$ is the set of lines in $\mathcal{L}$ that connected to bus $i$.
 We suppose the attacker uses multiple linear regression to learn $[\hat{f}_{0}\; \hat{F}]$. 

Multiple linear regression is a statistical method to find a linear relationship between multiple inputs and single output \cite{RegressionBook}. Take boundary bus $i$ for an example. Let the output $y_i=\bar{P}_{I,i}$ and inputs $\mathbf{x}^T = \bar{G}\bar{P}_G-\bar{P}_D$. At one instance of time $t$, $y_{i,t}$ satisfies
\vspace{-0.2cm}
\begin{equation}
y_{i,t} = \left[\begin{array}{cccc} x_{1,t} & x_{2,t} & \hdots & x_{k,t}\end{array} \right] \left[\begin{array}{c} f_{i,1} \\ f_{i,2}\\ \vdots\\ f_{i,k} \end{array}\right]+ f_{i,0} + \varepsilon_{i,t} \label{eq:MLP_variable}
\end{equation}
where $f_{i,j}$, $j=0,...,k$, are regression coefficients for boundary bus $i$, and $\varepsilon_{i,t}$ is random error. In the following, we let $\hat{F}_i=\left[f_{i,1}\;\;f_{i,2}\;...\; f_{i,k}\right]$ be the coefficient vector for boundary bus $i$.

Consider a problem with an $m\times 1$ observed output vector $\mathbf{y}_i$, and an $m\times k+1$ input matrix $\mathbf{X}=[\mathbf{1} \;\mathbf{x}_1 \;...\; \mathbf{x}_k]$. The relationship in \eqref{eq:MLP_variable} can be written in matrix notation as
\begin{equation}
\mathbf{y}_i=\mathbf{X}\left[\hat{f}_{i,0} \;\; \hat{F}_i\right]^T + \mathbf{\varepsilon}_i. \label{eq:MLP_matrix}
\end{equation}
Least squares estimation (LSE) can be used to estimate the regression coefficients $\hat{F}_i$ in \eqref{eq:MLP_matrix} as 
\begin{equation}
[\hat{f}_{i,0} \;\; \hat{F}_i]^T=\left(\mathbf{X}^T\mathbf{X}\right)^{-1}\mathbf{X}^T\mathbf{y}_i. \label{eq:LSE_coefficient}
\end{equation}
Note that as we stated in K2, we assume the attacker has enough historical data. Therefore, $\left(\mathbf{X}^T\mathbf{X}\right)$ is full-rank. We repeatedly use this process to obtain $[\hat{f}_{i,0} \; \hat{F}_i]$ for each $i$, $i\in\mathcal{B}$. Thus, the attacker can use historical data to obtain the estimate $\hat{F}$, such that $[\hat{f}_0\;\hat{F}]=[\hat{f}_{i,0} \; \hat{F}_i]$, $\forall i\in \mathcal{B}$. The dimension of $[\hat{f}_0 \; \hat{F}]$ is $n_{\mathcal{B}}\times (k+1)$, where $n_{\mathcal{B}}$ is the number of boundary buses.
The attacker now can approximate \eqref{eq:LocalPF_b} as 
\begin{equation}
\bar{P}=\bar{K}\left(\bar{G}\bar{P}_G-\bar{P}_D\right)- \bar{K}^{\mathcal{B}}\left(\hat{F}\left(\bar{G} \bar{P}_G - \bar{P}_D\right)+\hat{f}_0\right). \label{eq:PI_predict}
\end{equation} 
In the following subsection, \eqref{eq:PI_predict} is used to evaluate the vulnerability to attacker with limited information. 
\begin{figure}[h]
	\centering{}\includegraphics[trim=0 1.2cm 0 1.2cm, scale=0.4]{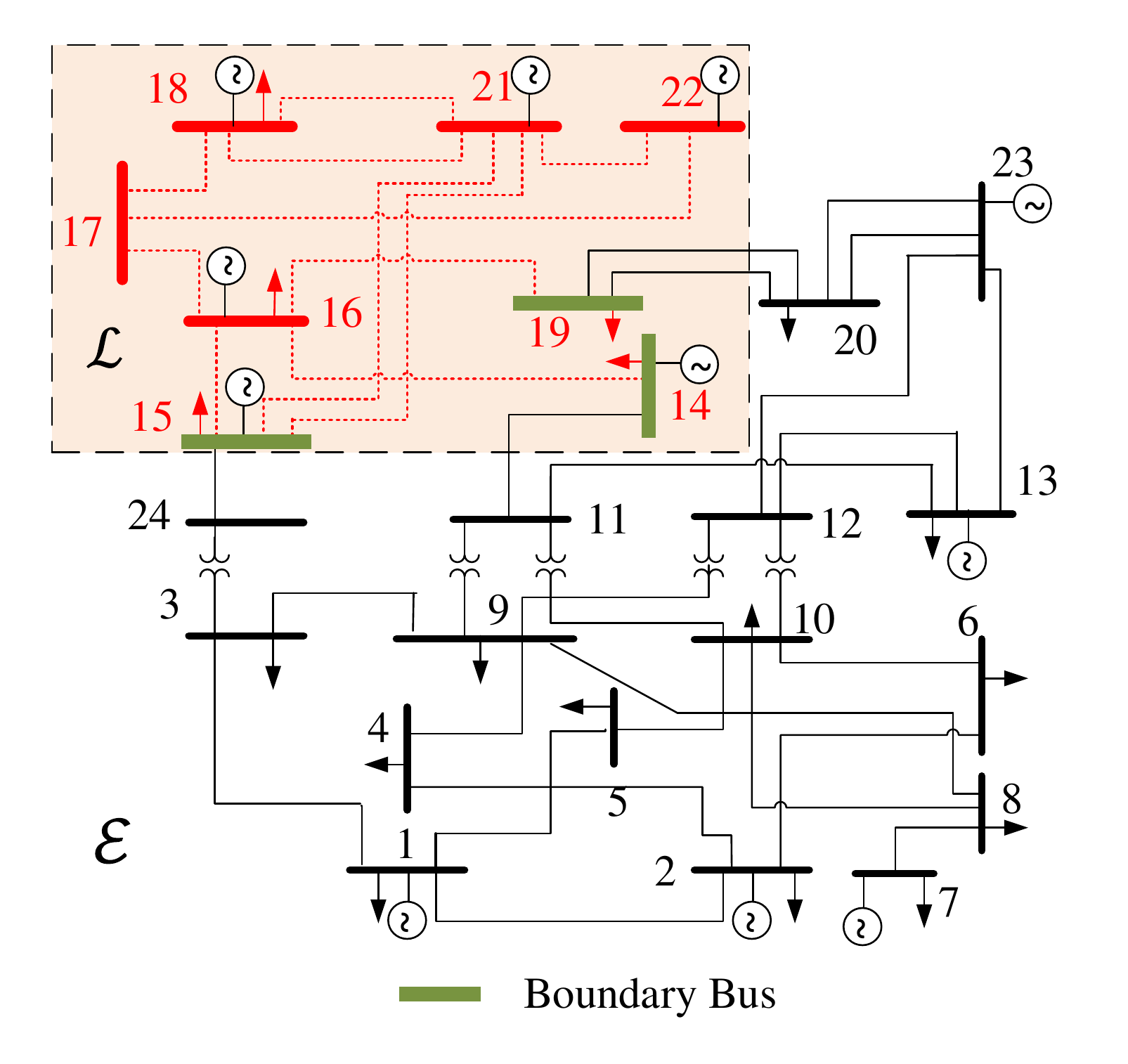}\protect\protect\caption{IEEE 24-bus RTS system decomposed into attack sub-network and attack external network. \label{fig:LocalAttack}}
\end{figure}

\subsection{\label{sub:LA_define}Attack Optimization Problem under Localized Information}
In this section, we introduce a bi-level attack optimization problem to formulate the limited information attack. The first level determines the attack vector in $\mathcal{L}$ that maximize target line flow and the second level represents system re-dispatch after attack via DC OPF formulated with only information in $\mathcal{L}$. However, since the attacker does not have knowledge of either the topology or the generator information in $\mathcal{E}$, we assume that the attacker only minimizes the total cost of generation in $\mathcal{L}$ and approximates the effect of the total generation cost in $\mathcal{E}$ as the total cost of the pseudo-boundary injections in the second level modified OPF. For boundary bus $i$, this cost is estimated as the product of the LMP, $\lambda_i$, and the pseudo-boundary injection at bus $i$. The limited information bi-level attack optimization problem is as follows:
\begin{flalign}
\underset{\bar{c},\bar{P}}{\text{maximize}} \;\;& \bar{P}_{l}-\zeta\left\Vert \bar{c}\right\Vert _{0}\label{eq:Obj1_MaxPF_1}\\
\notag \text{subject to}\;\;\\
&\hspace{-1cm} \bar{P}=\bar{K}\left(\bar{G}\bar{P}_G^{*}-\bar{P}_D\right)- \bar{K}^{\mathcal{B}}\bar{P}_{I,\mathcal{B}}^* \label{eq:Physical_PF_1}\\
&\hspace{-0.5cm} \left\Vert \bar{c}\right\Vert _{0}\leq N_{0}, \hspace{0.5cm}  \bar{c}_{\mathcal{B}}=\mathbf{0}\label{eq:con_resources_1}\\
&\hspace{-0.5cm} -\tau  \bar{P}_{D}\leq \bar{H}\bar{c}\leq\tau \bar{P}_{D} \label{eq:con_loadshift_1}\\
&\hspace{-1.8cm} \left\{ \bar{P}_{G}^{*},\bar{P}_{I,\mathcal{B}}^*\right\}=\text{arg}\left\{ \hspace{-0.1cm} \underset{\bar{P}_{G},\bar{P}_{I,\mathcal{B}}}{\text{minimize}}\negthinspace\underset{g\in \mathcal{L}}{{\sum}} C_{g}\left(\bar{P}_{Gg}\right)+\sum\limits_{i\in \mathcal{B}} \lambda_i \bar{P}_{I,i}\right\} \label{eq:OBJ_MINCOST_1} \\
\notag \text{subject to}\hspace{0.2cm}\;\\
& \hspace{0.2cm} \bar{P}_{I,\mathcal{B}} = \hat{F}\left(\bar{G}\bar{P}_G-\bar{P}_D+\bar{H}\bar{c}\right)+\hat{f}_0 \label{eq:PseudoInjection_opt}\\
&  \hspace{0.2cm}\sum\limits_{g\in {\mathcal{L}}} \bar{P}_{G,g} - \sum\limits_{i\in\mathcal{B}}\bar{P}_{I,i}= \sum\limits_{i\in \mathcal{L}} \bar{P}_{D,i}\\
&\hspace{-1.4cm} -\bar{P}_{\text{max}}\leq \bar{K}\left(\bar{G} \bar{P}_G - \bar{P}_D+\bar{H}\bar{c}\right) - \bar{K}^{\mathcal{B}}\bar{P}_{I,\mathcal{B}} \leq \bar{P}_{\text{max}}\label{eq:ThermalLimit_Local}\\
&\hspace{0.5cm}\bar{P}_{G,\text{min}}\leq \bar{P}_G \leq \bar{P}_{G,\text{max}} \label{eq:PG_limit}
\end{flalign}
where \eqref{eq:OBJ_MINCOST_1} captures the modified OPF objective as the first term represents the total cost of generation in $\mathcal{L}$ and the second term is the total cost of pseudo-boundary injections. Constraint \eqref{eq:PseudoInjection_opt} represents the attacker's prediction of the pseudo-boundary injection after attack resulting from the counterfeit loads. Note that, in the cyber system (OPF with attack vector), the power injections in $\mathcal{L}$ is $\bar{G}\bar{P}_G-\bar{P}_D+\bar{H}\bar{c}$, thus, the corresponding pseudo-boundary injection should respond to these injections with attack. In \eqref{eq:ThermalLimit_Local}, we directly write the second term with $K^{\mathcal{B}}\bar{P}_{I,\mathcal{B}}$ instead of $K^{\mathcal{B}}\left(\hat{F}\left(\bar{G}\bar{P}_G-\bar{P}_D+\bar{H}c\right)+\hat{f}_0\right)$. In addition, we have changed the constraint on the attack vector in \eqref{eq:con_resources_1} to limit the attack to be within the sub-network $\mathcal{L}$. 

As with the bi-level optimization problem for perfect information,  \eqref{eq:Obj1_MaxPF_1}$-$\eqref{eq:PG_limit} is non-linear and non-convex. We employ the same modifications as detailed in Sec. \ref{sub:GlobalAttack} to convert it into a MILP.


Note that attacker can only overload lines in $\mathcal{L}$.
	The attack optimization problem ensures that only measurements inside $\mathcal{L}$ can be changed by attacker. The post-attack system re-dispatch (OPF), on the other side, forces all the cyber line power flows within the thermal limits. Therefore, the attacker can only hide the physical overflow inside $\mathcal{L}$ with FDI attack.



\section{\label{sec:validation} Justification of the Localized Information FDI Attacks}
In this section, we make a distinction between the \emph{physical} system, as it actually exists, and the \emph{cyber} system, as seen by the control center, which may differ from the physical system due to the FDI attack. We use the superscripts $p$ and $c$ to denote the physical and cyber power flows, respectively. Due to limited information, the attacker can only use data in $\mathcal{L}$ to compute the physical and cyber power flows which may be different from the actual values. Therefore, we refer the physical and cyber power flows computed by the attacker as \textit{attacker-computed} physical and cyber power flows, respectively.

We prove that: (i) there exists a linear relationship $F$ between pseudo-boundary injection and power injections in $\mathcal{L}$ under certain circumstances; and (ii) even if $\hat{F}$ does not accurately predict the system response after attack, the attacker can still compute an upper bound on the physical power flow with limited information.


The following assumptions are made about the historical data available to the attacker: (i) the topology for all the historical data remains the same, (ii) each instance of historical data satisfies OPF, and (iii) there exists a subset of buses $\mathcal{Z}$ in $\mathcal{E}$, for which power injections remain constant in the historical data. The subset of remaining buses in $\mathcal{E}$ is denoted as $\mathcal{Y}=\mathcal{E}\setminus\mathcal{Z}$. In our prior work \cite{Liang2015,Zhang2016TSG,ZhangPES2016,Chu2016SmartGridComm}, we have shown that congested lines are more vulnerable to line overflow FDI attacks. Analogously, in this work, we assume the target line is congested.

\subsection{Validation of Multiple Linear Regression Method} \label{sub:Validation_MLR}
In this subsection, we prove the existence of linear relationship between pseudo-boundary injections and power injections in $\mathcal{L}$ under certain circumstances.


For simplicity, we define the set of lines in the network $\mathcal{G}$ that are the congested for each instance of historical data as $\mathbb{C}$, where $\mathbb{C}^+$ and $\mathbb{C}^-$ are the subsets in $\mathbb{C}$ for which the power flow directions are positive and negative, respectively. We assume there are $n_c$ congested lines in $\mathbb{C}$, $n_{\mathcal{L}}$, $n_{\mathcal{E}}$, and $n_{\mathcal{Y}}$ buses in $\mathcal{L}$, $\mathcal{E}$, and $\mathcal{Y}$, respectively. 

In order to evaluate the performance of the coefficient matrix $\hat{F}$, we define a matrix $B=[K_{\mathbb{C}}^{\mathcal{Y}}; \; \mathbf{1}^T]$, where $K_{\mathbb{C}}^{\mathcal{Y}}$ is the sub-matrix of $K$ whose rows correspond to the congested lines in $\mathbb{C}$ and columns correspond to the buses in $\mathcal{Y}$.
\begin{theorem}
	The coefficient matrix $\hat{F}$ perfectly predicts the pseudo-boundary injections with power injections in $\mathcal{L}$ linearly if and only if $B$ is full column rank.
\end{theorem} 	

\begin{proof}
We denote the vector of power injections in $\mathcal{G}$ as $v$; that is $v=GP_G-P_D$; the vectors $v_{\mathcal{L}}$, $v_{\mathcal{E}}$, $v_{\mathcal{Y}}$, and $v_{\mathcal{Z}}$ represent the subsets of $v$ corresponding to buses in $\mathcal{L}$, $\mathcal{E}$, $\mathcal{Y}$, and $\mathcal{Z}$, respectively.  We define $\mathbb{W}_i$ as the set of lines connecting to boundary bus $i$, $i\in\mathcal{B}$, where $\mathbb{W}_i^{\mathcal{L}}$ and $\mathbb{W}_i^{\mathcal{E}}$ are the subsets of lines located in $\mathcal{L}$ and $\mathcal{E}$, respectively. We define a vector $J_i$ as the sum of row vectors in $K$ corresponding to lines in $\mathbb{W}_i^{\mathcal{E}}$; that is $J_i=\sum_{k\in\mathbb{W}^{\mathcal{E}}_i}K_k$. The matrices $J_i^{\mathcal{L}}$, $J_i^{\mathcal{Y}}$, and $J_i^{\mathcal{Z}}$ are the sub-matrices of $J$ in which the columns of the matrices corresponding to buses in $\mathcal{L}$, $\mathcal{Y}$, and $\mathcal{Z}$, respectively. As introduced in Sec. \ref{sec:LocalAttack}, the pseudo-boundary injection at bus $i$ is a linear combination of power injections at each bus in $\mathcal{L}$, $\mathcal{Y}$, and $\mathcal{Z}$ is given by
\begin{equation}
\bar{P}_{I,i} = J_i^{\mathcal{L}} v_{\mathcal{L}}+J_i^{\mathcal{Y}} v_{\mathcal{Y}} + J_i^{\mathcal{Z}} v_{\mathcal{Z}}. \label{eq:PINJ_wholeNetwork}
\end{equation}
Note that $v_{\mathcal{Z}}$ is a constant across all instances of historical data.
Since each instance of historical data resulted from an converged OPF, $v_\mathcal{L}$ and $v_\mathcal{Y}$ satisfy the following: 
\begin{flalign}
K_k^{\mathcal{L}}v_{\mathcal{L}} + K_k^{\mathcal{Y}} v_{\mathcal{Y}} = P_{k,\text{max}}-K_k^{\mathcal{Z}} v_{\mathcal{Z}} \;\;\;\;\;\; \forall k\in \mathbb{C}^+ \label{eq:con_br+}\\
K_r^{\mathcal{L}}v_{\mathcal{L}} + K_r^{\mathcal{Y}} v_{\mathcal{Y}} = -P_{r,\text{max}}-K_r^{\mathcal{Z}} v_{\mathcal{Z}} \:\;\; \; \forall r\in \mathbb{C}^-\label{eq:con_br-}\\
\mathbf{1}^T v_{\mathcal{L}} + \mathbf{1}^T v_{\mathcal{Y}} = - \mathbf{1}^T v_{\mathcal{Z}} \hspace{3.1cm}  \label{eq:P_B}
\end{flalign}
where \eqref{eq:con_br+} and \eqref{eq:con_br-} are the thermal limit constraints for congested lines in $\mathbb{C}^+$ and $\mathbb{C}^-$, respectively, and \eqref{eq:P_B} is the power balance constraint.

Equations \eqref{eq:con_br+}$-$\eqref{eq:P_B} can be collected as
\begin{equation}
Av_{\mathcal{L}} + B v_{\mathcal{Y}} = d. \label{eq:ABF}
\end{equation}
where $A=[K_{\mathbb{C}}^{\mathcal{L}}; \; \mathbf{1}^T]$ and $d=[SP_{\mathbb{C},\text{max}}-K_{\mathbb{C}}^{\mathcal{Z}}v_{\mathcal{Z}}; \;-\mathbf{1}^T v_{\mathcal{Z}}]$. The matrix $S$ is a $n_c \times n_c$ diagonal matrix with $S_{kk}=1$, $\forall k\in\mathbb{C}^+$, and $S_{kk}=-1$, $\forall k\in \mathbb{C}^-$.

The dimensions of $A$ and $B$ are $\left(n_c+1\right) \times n_{\mathcal{L}}$ and $\left(n_c+1\right) \times n_{\mathcal{Y}}$, respectively. Note that the number of columns in $B$ represents the total number of buses in $\mathcal{Y}$.



Suppose that $B$ is full column rank. Thus, $B^TB$ is non-singular; that is, there exists a pseudoinverse $B^+=\left(B^TB\right)^{-1}B^T$, such that $B^+B=I$. Therefore, applying $B^+$ to \eqref{eq:ABF}, the vector $v_{\mathcal{Y}}$ can be rewritten as
	\begin{equation}
	v_{\mathcal{Y}} = -B^+Av_{\mathcal{L}} + B^+d. \label{eq:SwithVL}
	\end{equation}
	The pseudo-boundary injection $\bar{P}_{I,i}$ in \eqref{eq:PINJ_wholeNetwork} can be written as
	\begin{equation}
	\bar{P}_{I,i} = \left(J_i^{\mathcal{L}}-J_i^{\mathcal{Y}}B^+A\right) v_{\mathcal{L}} + J_i^{\mathcal{Y}}B^+d +J_i^{\mathcal{Z}}v_{\mathcal{Z}}. \label{eq:Psij}
	\end{equation}
	Therefore, the linear coefficient $F_i$ between $\bar{P}_{I,i}$ and $v_{\mathcal{L}}$ is 
	\begin{equation}
	\begin{array}{c}
	F_i=J_i^{\mathcal{L}}-J_i^{\mathcal{Y}}B^+A\\
	f_{i,0}=J_i^{\mathcal{Y}}B^+d +J_i^{\mathcal{Z}}v_{\mathcal{Z}}.\end{array} \label{eq:Coeff_F_SecV}
	\end{equation}
From \eqref{eq:Coeff_F_SecV}, we see that $F_i$ is unique and is the perfect linear predictor. The linear coefficient matrix between $\bar{P}_{I,\mathcal{B}}$ and $v_{\mathcal{L}}$ is $F=\left[F_i\right]$, $\forall i\in\mathcal{B}$.

	 Suppose that $B$ is not full column rank. Thus there exist infinitely many of $v_\mathcal{Y}$ satisfying \eqref{eq:ABF}, \textit{i.e.}, $v_{\mathcal{Y}}$ cannot be uniquely determined by $v_{\mathcal{L}}$. Therefore, the multiple linear regression will not perfectly predict the pseudo-boundary injections.
	 \end{proof}
	 
	 In Sec. \ref{sec:Simulation}, we provide a test case in IEEE 24-bus system for which $B$ is full column rank. We demonstrate that the $\hat{F}$ obtained with multiple linear regression method does indeed lead to perfect prediction of $\bar{P}_{I,\mathcal{B}}$. We also provide three counter examples (one in IEEE 24-bus system and the others in IEEE 118-bus system). For these illustrated counter-examples, $B$ satisfies $\left(n_c+1\right)< n_{\mathcal{Y}}$, which indicates that $B$ is not full column rank. However, even for a case with $B$ satisfying $\left(n_c+1\right)\geq n_{\mathcal{Y}}$, $B$ cannot be assumed to be full column rank. An example is a system with 3 buses in $\mathcal{Y}$ and 2 parallel congested lines. For this system, $\text{rank}\left(K_{\mathbb{C}}^{\mathcal{Y}}\right)=1$ since the row vectors in $K$ for the parallel lines are the same. The matrix $B$, hence, is not a full rank matrix since $\text{rank}\left(B\right)\leq 2$ and by Theorem 1, $\hat{F}$ cannot result in an accurate prediction.    

Note that $B$ does not determine the feasibility of the limited information FDI attacks. In fact, $B$ only determines whether  $\bar{P}_{I,\mathcal{B}}$ can be perfectly predicted by $v_\mathcal{L}$ or not. However, that does not mean that when $B$ is not full column rank, such attacks are infeasible. The matrix $B$ which is not full column rank may undermine the attacker's evaluation of the attack consequences via the bi-level attack optimization problem. But the attacker can still find attack vector $\bar{c}^*$ and design the attack.


\subsection{\label{sub:Validation_OPT} Upper Bound on Physical Consequences of Attack}
Although $\hat{F}$ in general cannot accurately predict $\bar{P}_{I,\mathcal{B}}$ when $B$ is not full column rank, the attacker can still utilize $\hat{F}$ in the bi-level attack optimization problem \eqref{eq:Obj1_MaxPF_1}$-$\eqref{eq:PG_limit} to predict the physical power flow on target line. However, the attacker-computed physical power flow may not match the physical power flow. The following theorem shows that even so, the attacker can compute an upper bound $P_l^{ub}$ on the physical power flow on the target line subsequent to an attack.

\begin{theorem} 
	The physical power flow on the target line $l$ resulting from attack vector $\bar{c}^*$ is upper bounded by
	\begin{equation}
	P_l^{\text{ub}}=P_{l,\text{max}}-\bar{K}_l\bar{H}\bar{c}^*. \label{eq:upperbound}
	\end{equation}
\end{theorem}

\begin{proof}
  Solving the attack optimization problem \eqref{eq:Obj1_MaxPF_1}$-$\eqref{eq:PG_limit}, the attacker can obtain the optimal attack vector $\bar{c}^*$. The resulting attack vector for the whole system is $c^*$, where $c_i^*=\bar{c}^*_i$ for $i\in\mathcal{L}$ and $c^*_i=0$ for $i \in \mathcal{E}$. Injecting $c^*$ in the system will result in a system re-dispatch determined by \eqref{eq:OBJ_MINCOST}$-$\eqref{eq:con_GENlimit}. The difference between the physical and cyber power flows ($P_l^p$ and $P_l^c$, respectively) on target line $l$ after the post-attack system re-dispatch is
 \begin{equation}
P_l^p - P_l^c = -K_l Hc^*. \label{eq:global_dPF}
 \end{equation} 
Thus, the physical power flow on target line $l$ satisfies
\begin{equation}
P_l^p = P_l^c -K_l  Hc^*\leq P_{l,\text{max}}-K_l  Hc^*. \label{eq:global_dPF1}
\end{equation} 
where the upper bound follows from the thermal limit constraint on $P_l^c$ in \eqref{eq:con_powerflow}.
Note that $K_l$ and $H$ are unknown to the attacker with limited information. However, the attacker has the knowledge of $\bar{K}_l$ and $\bar{H}$. We now show that the upper bound in \eqref{eq:global_dPF1} is equivalent to $P_l^{\text{ub}}$ defined in \eqref{eq:upperbound}.

The PTDF matrices $\bar{K}$ and $K$ satisfy the following
\begin{equation}
\bar{K} = \bar{\Gamma}\bar{H}^+ \label{eq:localPTDF_H2}
\end{equation}
\begin{equation}
K = \Gamma H^+ \label{eq:PTDF_H2}
\end{equation}
where $\Gamma$ and $\bar{\Gamma}$ are the dependency matrices between power flow measurements and voltage angle states in $\mathcal{G}$ and $\mathcal{L}$, respectively, $H^+$ and $\bar{H}^+$ are the pseudoinverse of $H$ and $\bar{H}$, respectively. Note that for target line $l$, both $\Gamma_l$ and $\bar{\Gamma}_l$ have only two non-zero elements corresponding to the two end buses of $l$ (denoted as buses $l_f$ and $l_t$, respectively). In particular, $\Gamma_l^{l_f}=\bar{\Gamma}_l^{l_f}=-\Gamma_l^{l_t}=-\bar{\Gamma}_l^{l_t}=\frac{1}{x_l}$, where $x_l$ is the line impedance of line $l$. Thus,
\begin{equation}
\bar{K}_l\bar{H}\bar{c}^* = \bar{\Gamma}_l\bar{c}^* = \Gamma_lc^* =  K_lHc. \label{eq:dPL_equal}
\end{equation}
Therefore, the right-hand side of \eqref{eq:global_dPF1} is exactly equal to $P_l^{\text{ub}}$. This proves the upper bound in \eqref{eq:upperbound}. Moreover, $P_l^{\text{ub}}$ can be computed by the attacker, since it requires knowledge only of the local network $\mathcal{L}$ and the attack vector $\bar{c}^*$.
\end{proof}

Note that 
from \eqref{eq:Physical_PF_1} and \eqref{eq:ThermalLimit_Local}, the attacker can compute the difference between the physical and cyber power flows on target line $l$ ($\bar{P}_l^p$ and $\bar{P}_l^c$, respectively) solved with limited information attack optimization as
\begin{equation}
\bar{P}_l^p - \bar{P}_l^c = -\bar{K}_l \bar{H}\bar{c}^*. \label{eq:local_dPF}
\end{equation}
Thus, the difference between physical and cyber power flows seen by the attacker and the system are the same
\begin{equation}
\bar{P}_l^p - \bar{P}_l^c=P_l^p-P_l^c. \label{eq:local_dPF1}
\end{equation}

\section{\label{sec:Simulation}Numerical Results}

In this section, we illustrate the efficacy of the attacks designed with
the method proposed in Sec. \ref{sec:LocalAttack}. To this end, we first compute the coefficient matrix with historical data using the multiple linear regression method. Subsequently, we solve the optimization problem to find the optimal attack vector $\bar{c}^*$ inside $\mathcal{L}$. Finally, we test the physical consequences of the attack vector $\bar{c}^*$ on the entire network $\mathcal{G}$. The
test systems include the IEEE 24-bus reliability test system (RTS) and the IEEE 118-bus system from MATPOWER v4.1. In particular, the line rating data for IEEE 118-bus system is adopted from \cite{IEEE118Rating}. 
The whole network DC OPF and limited information attack algorithm is implemented with Matlab. The
optimization problem is solved with CPLEX.

To model realistic power systems, we assume
that there are congested lines prior to the attack and the attacker
chooses one of them in $\mathcal{L}$ as the target to maximize power flow. This is achieved in simulation by uniformly reducing all line ratings by 50\% for IEEE 24-bus RTS and 45\% for IEEE 118-bus system. 

We illustrate our results for the following choice of parameters: 
the weight of the $l_{1}$-norm of attack vector in \eqref{eq:Obj1_MaxPF_1}, $\zeta$, is set to 1\% of the original power flow value of the target line; and the load shift factor in \eqref{eq:con_loadshift_1}, $\tau$, is set to 10\%. We assume that the attacker can obtain 200 instances of historical data inside $\mathcal{L}$. 

We focus on two scenarios for the historical data: 
\begin{itemize}
	\item \textbf{Scenario 1 - Constant Loads in $\mathcal{E}$:} In each instance of data, loads in $\mathcal{E}$ remain unchanged while loads in $\mathcal{L}$ varies as a percent $p$ of the base load, where $p$ is independent $\mathcal{N}(0, 10\%)$. That is, power injections vary only at buses with marginal generators (denoted $\mathcal{E}_M$). Therefore, in this scenario, the set $\mathcal{Y}$ introduced in Sec. \ref{sec:validation} is given by $\mathcal{Y}=\mathcal{E}_M$. The number of buses in $\mathcal{E}_M$ is denoted by $n_{\mathcal{E}_M}$
	\item \textbf{Scenario 2 - Varying Loads in the entire network $\mathcal{G}$:} In each instance of data, loads in both $\mathcal{L}$ and $\mathcal{E}$ vary as a percent $p$ of the base load, with $p$ chosen independently for each load as $\mathcal{N}(0, 10\%)$. In this scenario, power injections at all buses in $\mathcal{E}$ vary in the historical data, \textit{i.e.,} $\mathcal{Y}=\mathcal{E}$. 
\end{itemize}
Note that the data in both scenarios also satisfy the following assumptions: (i) the topology for all the historical data remains the same, (ii) the historical generation dispatches data in both scenarios satisfies OPF, and (iii) there exists a subset of buses $\mathcal{Z}$ in $\mathcal{E}$, for which power injections remain constant in the historical data. An example of attack is that an attacker hacks into the system and collects the data from 12:00 p.m. to 2:00 p.m. for the entire month of July and then launches a FDI attack at the end of the month. 
\begin{table*}[t]
	\renewcommand{\arraystretch}{1.3}
	\protect\caption{Summary of the Test Systems\label{tab:TestComparison}}
	\centering
	\vspace{-0.2cm}
	\begin{tabular}{>{\centering}m{1.5cm}|>{\centering}m{1.8cm}|>{\centering}m{1.6cm}|>{\centering}m{1cm}|>{\centering}m{1.8cm}|>{\centering}m{1.5cm}|>{\centering}m{1cm}}
		\thickhline
		\multirow{2}{*}{} & \multicolumn{3}{c|}{Scenario 1} & \multicolumn{3}{c}{Scenario 2} \tabularnewline
		\hhline{~------}
		Test System & \# of Congested Lines ($n_c$)& \# of Buses in $\mathcal{E}_M$ ($n_{\mathcal{E}_M}$)& $\text{rank}(B)$ & \# of Congested Lines ($n_c$) & \# of Buses in $\mathcal{E}$ ($n_{\mathcal{E}}$) &$\text{rank}(B)$ \tabularnewline
		\hline
		24-bus & 2 & 2 & 2  & 2 & 16 & 3   \tabularnewline
		\hline 
		118-bus & 3 & 5 & 4  & 3 & 71 & 4   \tabularnewline
		\thickhline 
	\end{tabular}
\end{table*}

\subsection{\label{sub:IEEE24}Results for IEEE 24-Bus RTS System}
In this subsection, we present attack consequences on the IEEE 24-bus RTS system for Scenarios 1 and 2. The sub-network $\mathcal{L}$ is illustrated in Fig. \ref{fig:LocalAttack}. In each scenario, we compare the attack consequences on target line 28 determined by the optimization problems for two cases: (i) complete system knowledge as in \cite{Liang2015} (identified as \textit{global case}), and (ii) limited system knowledge (henceforth identified as \textit{local case}). For local case, we compare the physical power flow $P_l^p$ and the attacker-computed physical power flow $\bar{P}_l^p$. The results of attacks are illustrated in Fig. \ref{fig:Sc1_Line28_IEEE24}. We illustrate the difference between the physical and the attacker-computed pseudo-boundary injections in Fig. \ref{fig:PBI_IEEE24}.

\begin{figure}[h]
	\centering{}\includegraphics[trim=0 0.3cm 0 0.5cm,scale=0.65]{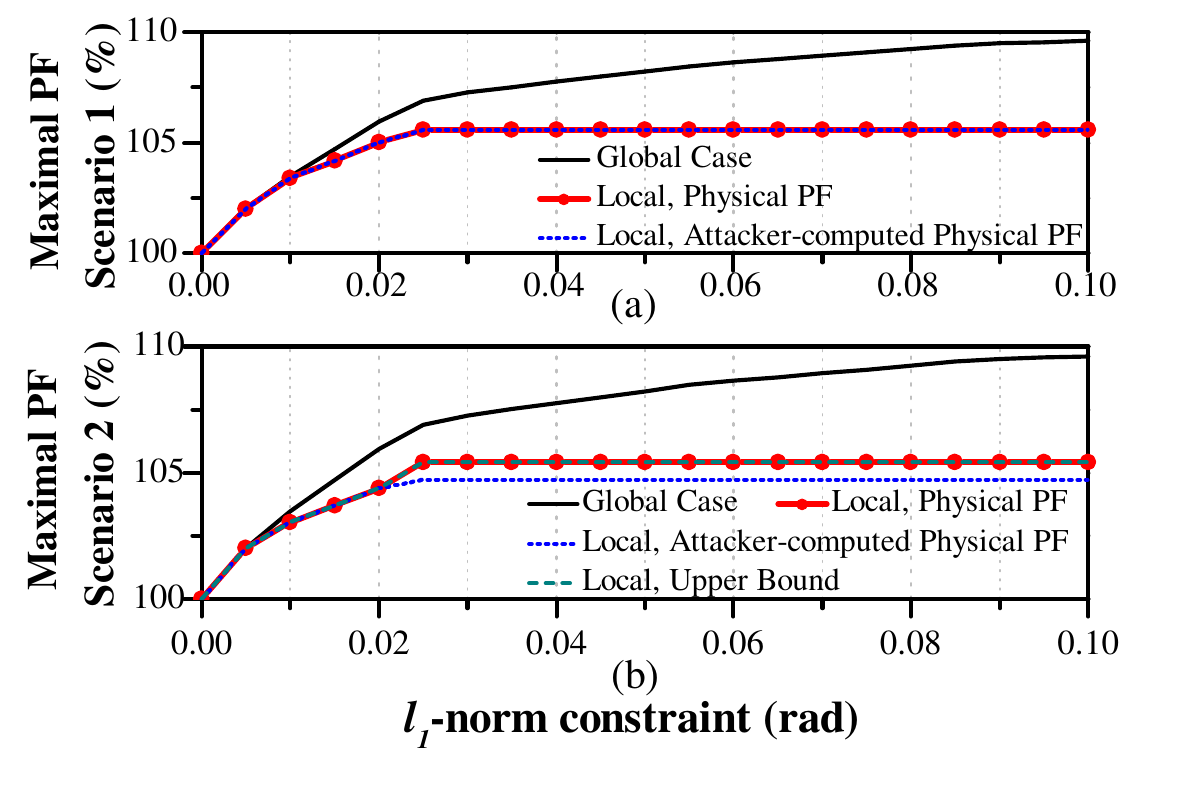}\protect\protect	\vspace{-0.3cm}\caption{The maximum power flow (PF) v.s. the $l_1$-norm constraint (N1) when
		target line is 28 of IEEE 24-bus system for (a) Scenario 1, and (b) Scenario 2 historical data. \label{fig:Sc1_Line28_IEEE24}}
\end{figure}
\begin{figure}[h]
	\centering{}\includegraphics[trim=0 0.3cm 0 0.6cm,scale=0.65]{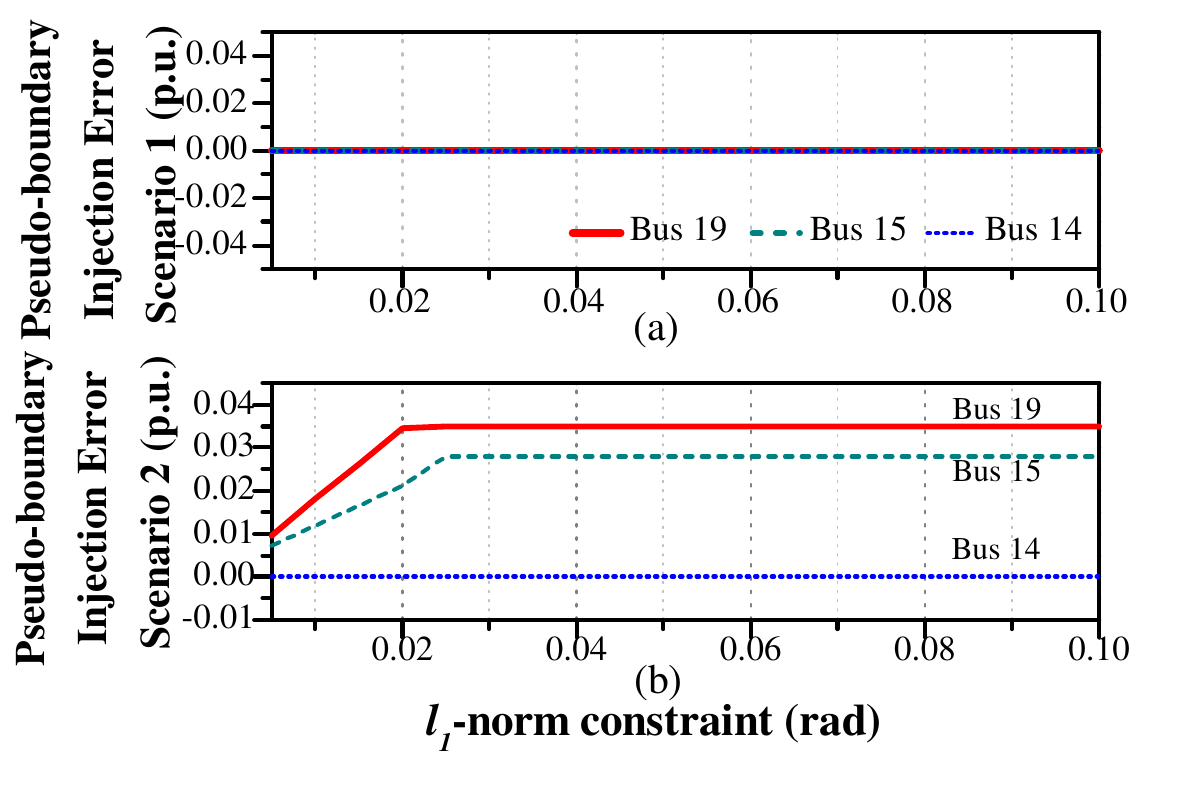}\protect\protect	\vspace{-0.3cm}\caption{The pseudo-boundary power injection error v.s. the $l_1$-norm constraint (N1) when
		target line is 28 of IEEE 24-bus system for (a) Scenario 1, and (b) Scenario 2 historical data. \label{fig:PBI_IEEE24}}
	\vspace{-0.2cm}
\end{figure}
In Figs. \ref{fig:Sc1_Line28_IEEE24}(a) and (b), we note that the solutions for the local case is sub-optimal relative to that for the global case. The reason is that as $N_{1}$ is relaxed, getting a larger overflow on the target line requires measurements in both $\mathcal{L}$ and $\mathcal{E}$ to be modified. Therefore, the constraint on limited attack resources prevents any further increase in the maximal target line flow for the local case. 

The parameters of the test system are summarized in Table \ref{tab:TestComparison}. The historical data in Scenario 1 satisfies $\text{rank}(B) = n_{\mathcal{E}_M}$. Thus, by Theorem 1, the pseudo-boundary power injections are perfectly predicted by the multiple linear regression method, which explains why the attacker-computed system response post-attack is the same as the actual response, as illustrated in Figs. \ref{fig:Sc1_Line28_IEEE24}(a) and \ref{fig:PBI_IEEE24}(a). 

Furthermore, Table \ref{tab:TestComparison} shows that for the historical data in Scenario 2, $\text{rank}(B)< n_{\mathcal{E}}$. Thus, by Theorem 1, the predictions of pseudo-boundary injections by the multiple linear regression are not accurate and there will be mismatches between the actual and the attacker-computed system response post-attack. This is verified by the non-zero pseudo-boundary power injection differences shown in Fig. \ref{fig:PBI_IEEE24}(b). In Fig. \ref{fig:Sc1_Line28_IEEE24}(b), in addition to plotting the attacker-computed physical power flow, we also plot the upper bound on physical power flow. From Fig. \ref{fig:Sc1_Line28_IEEE24}(b), we observe that although there are mismatches between the actual and attacker-computed system response under attack, the upper bound found in Sec. \ref{sub:Validation_OPT} exactly matches the physical power flow.     


\subsection{\label{sub:IEEE118}Results for IEEE 118-Bus System}
 
In this subsection, we test the consequences of attacks on the IEEE 118-bus system. The details of sub-network $\mathcal{L}$ are listed in Table \ref{tab:118sub-network}. The results of attacks designed with historical data in Scenarios 1 and 2 are illustrated in Fig. \ref{fig:Sc1_Line5_IEEE118} with sub-plots (a) and (b), respectively. The difference between the physical and the attacker-computed pseudo-boundary injections at 3 of 18 boundary buses, buses 23, 70, and 80, for both scenarios are illustrated in Fig. \ref{fig:PBI_IEEE118} with sub-plots (a) and (b), respectively.
\begin{figure}[h]
	\centering{}\includegraphics[trim=0 0.7cm 0 0.5cm,scale=0.65]{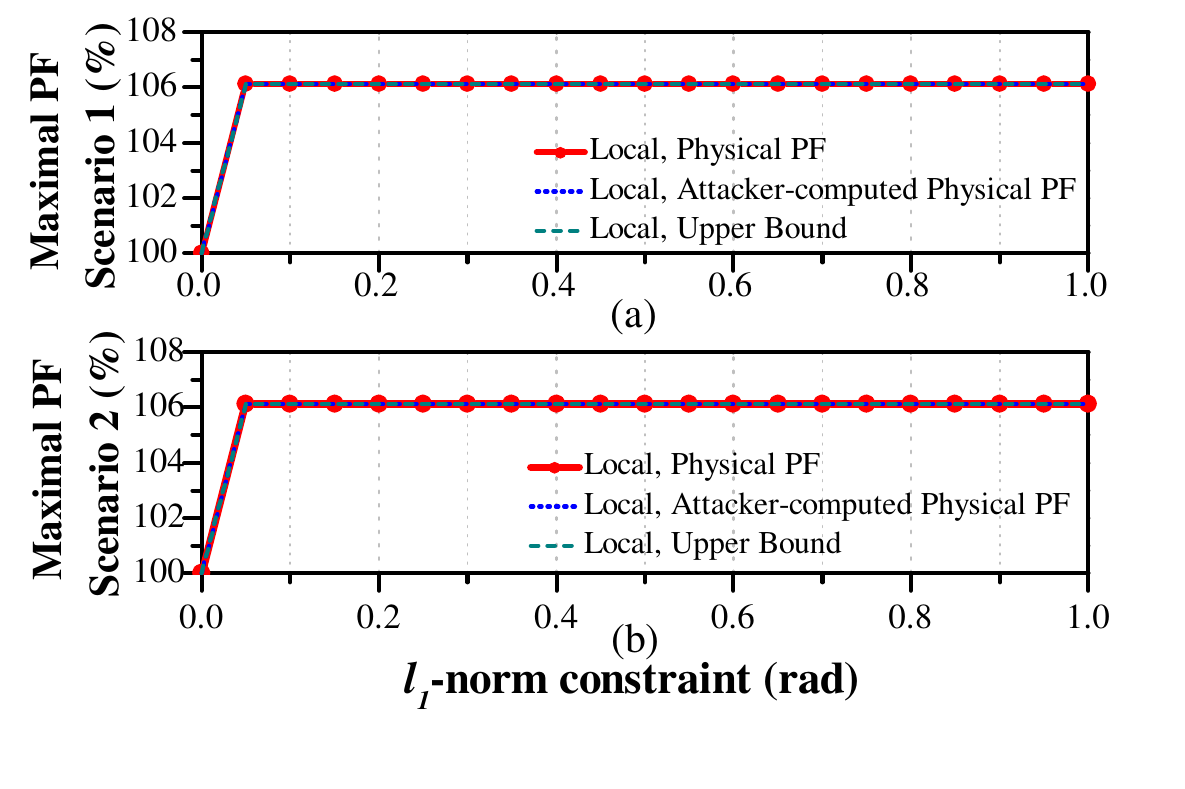}\protect\protect	\vspace{-0.3cm}\caption{The maximum power flow (PF) v.s. the $l_1$-norm constraint (N1) when
		target line is 5 of IEEE 118-bus system for (a) Scenario 1, and (b) Scenario 2 historical data. \label{fig:Sc1_Line5_IEEE118}}
\end{figure}
\begin{figure}[h]
	\centering{}\includegraphics[trim=0 0 0 0.5cm,scale=0.65]{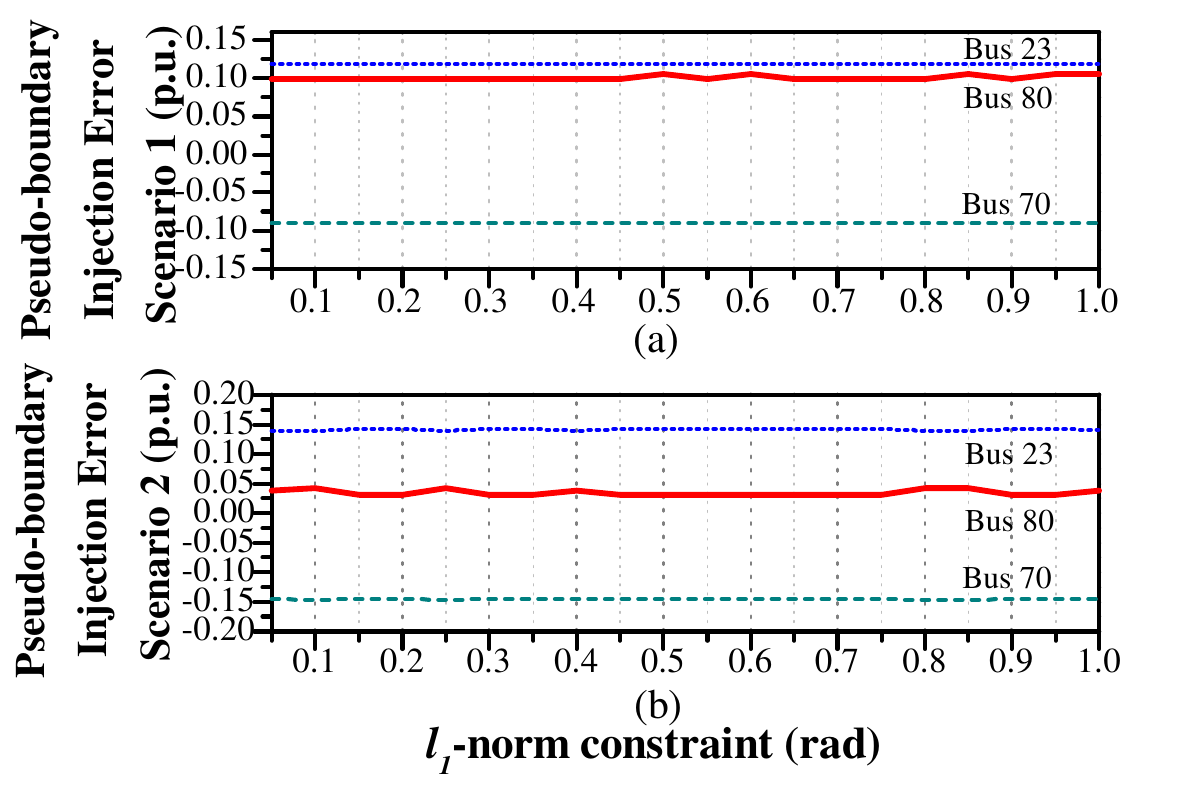}\protect\protect	\vspace{-0.3cm}\caption{The pseudo-boundary power injection error v.s. the $l_1$-norm constraint (N1) when
		target line is 5 of IEEE 118-bus system for (a) Scenario 1, and (b) Scenario 2 historical data. \label{fig:PBI_IEEE118}}
\end{figure}
The parameters of the test system are also summarized in Table \ref{tab:TestComparison}. Note that for historical data in both scenarios, $B$ does not have full column rank. Therefore, Theorem 1 predicts a mismatch between physical and attacker-computed pseudo-boundary injections. This is verified by Fig. \ref{fig:PBI_IEEE118}, which shows the pseudo-boundary injection error. In Figs. \ref{fig:Sc1_Line5_IEEE118}(a) and (b), we find that in both scenarios, both the attacker-computed physical power flow and the upper bound match the physical power flow. 
This case demonstrates that even though there are mismatches between physical and attacker-computed pseudo-boundary injections, the attacker-computed physical power flow can still be correct. Note that, in this case, both the cyber power flow and the attacker-computed cyber power flow reach the limit post-attack since the target line is congested before attack. Therefore, from \eqref{eq:local_dPF1}, the attacker-computed physical power flow is the same as the physical power flow.
	\begin{table}[h]
		\renewcommand{\arraystretch}{1.3}
		\protect\caption{Summary of the Attack Sub-network in IEEE 118-bus System\label{tab:118sub-network}}
		\centering
		\begin{tabular}{|>{\centering}p{1.2cm}|>{\centering}p{6.8cm}|}
			\hline
			Buses & 1-14, 16, 17, 23, 25-27,	30,	33-35,	37-40,	47,	49,	59-66,	68-70,	75,	77,	80,	81,	116, 117 \tabularnewline
			\hline
			Lines & 1-17, 20, 22, 31-33,	36-38,	47,	48,	50-55,	65,	88-100,	102,	104-108,	115,	116,	119,	120,	123,	124,	126,	127,	183, 184 \tabularnewline
			\hline 
			Boundary Buses & 13,	14, 17,	23,	27,	33-35,	40,	47,	49,	59,	62,	66,	70,	75,	77,	80\tabularnewline
			\hline 
		\end{tabular}
	\end{table}

\subsection{Attack Sensitivity to Topology Change}	
In this subsection, we evaluate the efficacy of attacks generated from historical data with topologies that are different from the real-time topology. We assume that the attacker uses the historical datasets in Secs. \ref{sub:IEEE24} and \ref{sub:IEEE118} to compute the coefficient matrices, and is not aware of a line outage in $\mathcal{E}$ in real-time.  We exhaustively test the consequences of the attacks designed with the computed coefficient matrices on all possible real-time topologies with one line outage in $\mathcal{E}$. Note that the topology changes that will result in infeasible pre-attack DC OPF solution are not considered here. The $l_1$-norm constraints ($N_1$) are chosen as $0.05$ and $0.4$ for the IEEE 24-bus and 118-bus systems, respectively. We compare the results with those in Secs. \ref{sub:IEEE24} and \ref{sub:IEEE118} and summarize them in Table. \ref{tab:SA_topology}. Note that for the pseudo-boundary power injection errors, we compare the $l_2$-norm of the errors on all boundary buses for each test case since there are multiple boundary buses in the test system. From the table, we can observe that if the attacker uses the coefficient matrix computed from historical data with different topology to design attacks, its evaluation of attack consequences may be undermined. Specifically for the case with Scenario 1 historical data in the IEEE 24-bus system, the attacker cannot obtain perfect prediction on pseudo-boundary injections any more since the real-time topology differs from that in the historical data.  However, for most of the test cases, \textit{i.e.,} 86.47\% and 98.26\% cases for the IEEE 24-bus and 118-bus systems, respectively, the attacker can still cause line overflow with the inaccurate coefficient matrices.

 \begin{table}[t]
 	\renewcommand{\arraystretch}{1.3}
 	\protect\caption{Summary of the Sensitivity Analysis Results under Different Topologies\label{tab:SA_topology}}
 	\centering
 	\vspace{-0.2cm}
 	\begin{tabular}{>{\centering}m{1.5cm}|>{\centering}m{1.2cm}|>{\centering}m{1.4cm}|>{\centering}m{1.4cm}|>{\centering}m{1.4cm}}
 		\thickhline
 		
 		Test System \& Scenario & \# of Total Test Cases & \% of Cases with Physical Overflow Decreases& \% of Cases without Physical Overflow & \% of Cases with Prediction Error Increases \tabularnewline
 		\hline
 		24-bus SC1 &  22 & 18.18\% & 13.63\% & 100\%    \tabularnewline
 		\hline 
 		24-bus SC2&  22 & 9.1\%  & 13.63\% & 90.91\%    \tabularnewline
 		\hline 
 		118-bus SC1&  115 &  0  & 0 & 63.48\%    \tabularnewline
 		\hline
 		118-bus SC2&  115 & 0  & 1.74\% & 10.43\%    \tabularnewline
 		\thickhline 
 	\end{tabular}
 	
 	SC: Scenario
 \end{table}

\subsection{Verification on AC Power Flow Model}
In this subsection, we test the performance of the proposed attack strategies on AC power flow model. We first compare the attack consequences of the DC attacks in Secs. \ref{sub:IEEE24} and \ref{sub:IEEE118} and the corresponding AC attacks in Fig. \ref{fig:DCtoAC}. The AC attacks are computed with the designed DC attack vector $c$ which satisfy \eqref{eq:DCtoAC}. The system re-dispatch in response to each AC attack is via AC OPF. This figure validates that although the attack vector is solved by a linear optimization problem,
	it can still cause overflows in the AC system and the AC attack consequences track those of the original DC attacks.
\begin{figure}
	\centering{}\includegraphics[scale=0.65]{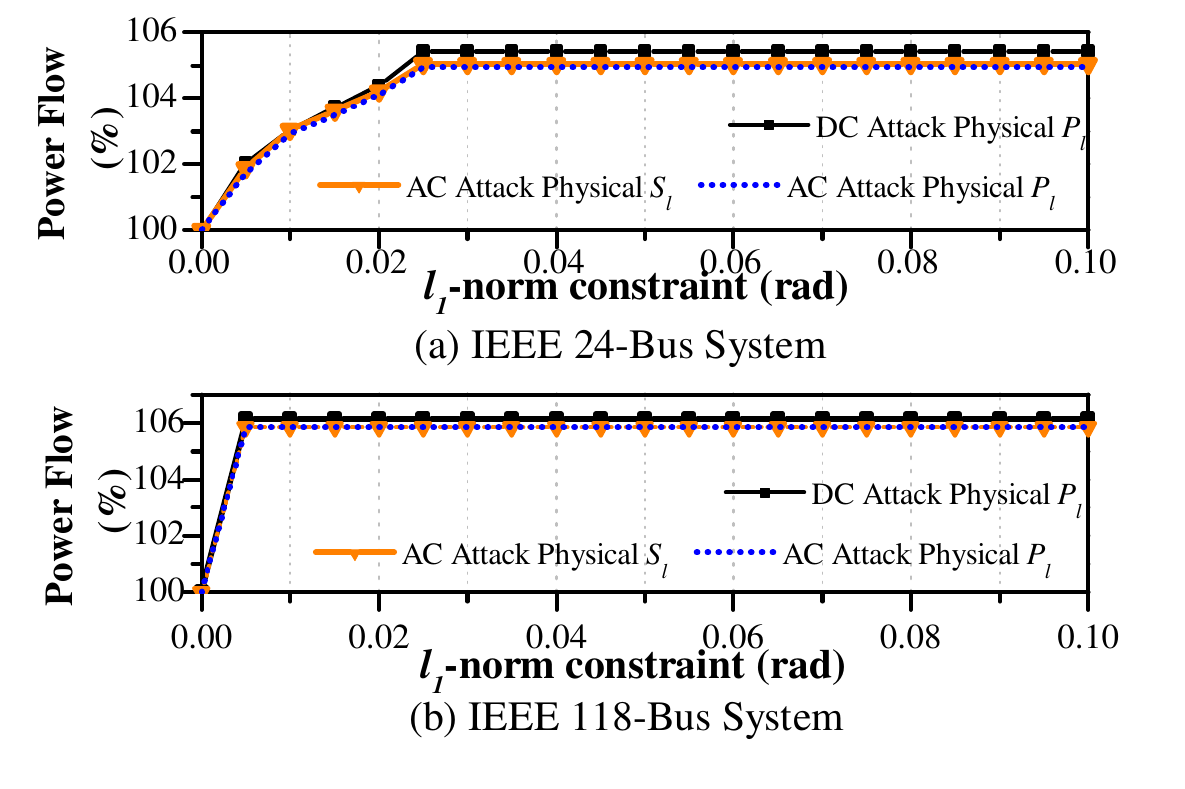}\protect\protect	\vspace{-0.3cm}\caption{Comparison of the maximum power flow of DC and AC attacks. \label{fig:DCtoAC}}
\end{figure}

In addition, the impact of AC power flow historical data on the attack consequences are studied. We randomly generate a historical dataset, in which each instance is based on AC power flow model and satisfies all assumptions in Scenario 2. That is, the topology in each instance of historical data is the same with the real-time topology. We then compute the coefficient matrix, solve the optimization problem to find the optimal attack, and test the physical consequences of the attack. The $l_1$-norm constraints ($N_1$) are chosen as $0.05$ and $0.4$ for the IEEE 24-bus and 118-bus systems, respectively. We repeat this process 100 times and illustrate the results in Table \ref{tab:SA_ACPF}. From this table, it can be seen that historical datasets with AC power flow data can reduce the prediction accuracy of the pseudo-boundary injections and the target line physical power flow. However, for both test systems, 100\% of the designed attacks can result in physical target line overflows. These results demonstrate the robustness of the proposed attack strategy on AC power flow model. 
 \begin{table}[h]
	\renewcommand{\arraystretch}{1.3}
	\protect\caption{Summary of the Sensitivity Analysis Results under AC Power Flow Historical Datasets\label{tab:SA_ACPF}}
	\centering
	\vspace{-0.2cm}
	\begin{tabular}{>{\centering}m{1.2cm}|>{\centering}m{1.4cm}|>{\centering}m{1.4cm}|>{\centering}m{1.4cm}}
		\thickhline
		Test System & \% of Cases with Physical Overflow & \% of Cases without Physical Overflow Decreases& \% of Cases with Prediction Error Increases \tabularnewline
		\hline
		24-bus&  100\% & 93\% & 46\%    \tabularnewline
		\hline 
		118-bus&  100\%  & 0 & 100\%    \tabularnewline
		\thickhline 
	\end{tabular}
\end{table}

\section{Conclusion}
In this paper, we have studied the physical system consequences of a class of unobservable limited information FDI attacks. We assume the attacker can design the attacks using historical data including topology, generation dispatch and load information only inside an attack sub-network $\mathcal{L}$ and can modify measurements within an attack sub-graph $\mathcal{S}$ inside $\mathcal{L}$ with counterfeits. The attacks are designed based on DC power flow model. We have introduced pseudo-boundary injections to represent the power flows delivered from the external network and developed a multiple linear regression model for the attacker to learn the relationship between pseudo-boundary injections and the power injections inside the attack sub-network. We have formulated a bi-level optimization
problem to maximize the power flow on a chosen target line with attacker's perfect information in the attack sub-network as well as the predicted pseudo-boundary injections. Our results illustrate that the attacker can overload transmission lines with the proposed bi-level attack optimization problems. In conclusion, one must be concerned that even with limited information, an attacker with access to historical data can take advantage of it to the detriment of reliable system operations. Future work will include (a) employing the attack optimization structure to achieve other attack consequences such as maximizing total operation costs, load shedding, or physical interface power flow that can result in voltage collapse, (b) performing detailed sensitivity analysis of this method with different scenarios of historical data, (c) extending this work to analyze the vulnerability of limited information cyber-physical topology-and-state attacks as introduced in our prior work \cite{Zhang2016TSG}, and (d) designing detection mechanisms to identify both the critical sub-graphs that are the most vulnerable to FDI attacks and anomalies inside these sub-graphs with machine learning algorithms.

\section*{Acknowledgment}
This material is based upon work supported by the National Science Foundation under Grant No. CNS-1449080.


\ifCLASSOPTIONcaptionsoff
  \newpage
\fi



%
\bibliographystyle{IEEEtran}
\bibliography{dis}

%
\begin{IEEEbiography}[{\includegraphics[width=1in,height=1.25in,clip,keepaspectratio]{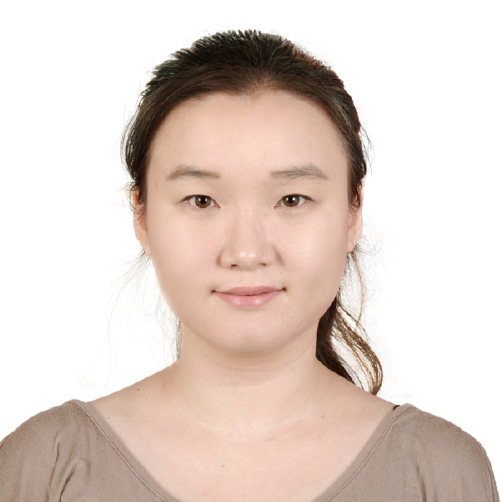}}]%
	{Jiazi Zhang} (S'14--M'17) received the B.S. degree from Shandong University, Jinan, China in 2012, the M. S. degree, and the Ph.D degree from Arizona State University, Tempe, AZ, USA in 2015 and 2017, respectively. 
	
	She is presently a Postdoctoral Researcher at National Renewable Energy Laboratory (NREL). Previously, she was an Assistant Postdoctoral Scholar at Arizona State University.
	Her research interests include cyber security of smart grid, optimization, energy market, and electrical energy storage.
\end{IEEEbiography}

\begin{IEEEbiography}[{\includegraphics[width=1in,height=1.25in,clip,keepaspectratio]{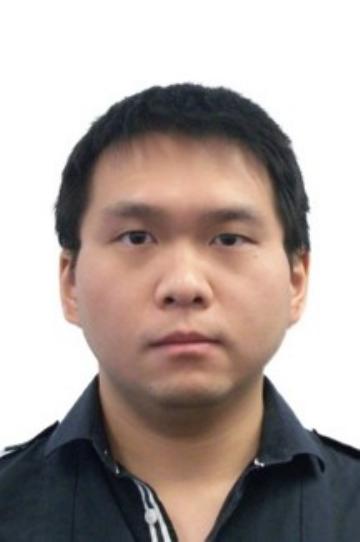}}]%
	{Zhigang Chu} (S'15) received the B.Eng. degree from Xi'an Jiaotong University, Xi'an, China and the M.S.E. degree from Arizona State University, Tempe, AZ, in 2012 and 2014, respectively.
	
	He is currently pursuing the Ph.D. degree in the School of Electrical, Computer, and Energy Engineering at Arizona State University. His research interests include cyber security of power systems, optimization, and energy markets.
\end{IEEEbiography}

\begin{IEEEbiography}[{\includegraphics[width=1in,height=1.25in,clip,keepaspectratio]{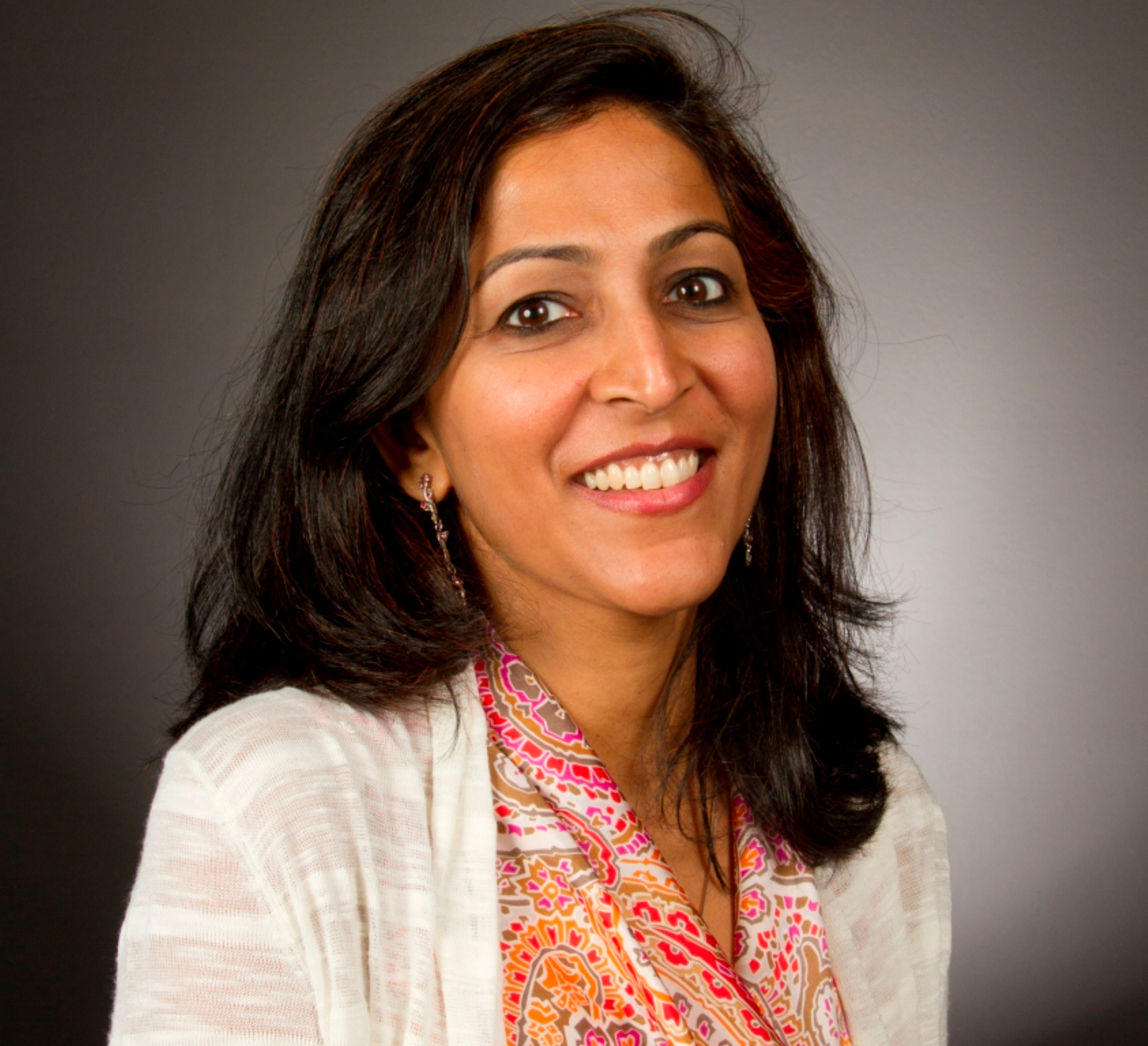}}]%
	{Lalitha Sankar} (S'02--M'07--SM'15) received the B.Tech degree from the Indian Institute of Technology, Bombay, the M.S. degree from the University of Maryland, and the Ph.D degree from Rutgers University. 
	
	She is presently an Assistant Professor in the School
	of Electrical, Computer, and Energy Engineering at Arizona State University. Previously, she was an Associate Research Scholar at Princeton University. Following her doctorate, Dr. Sankar was a recipient of a three year Science and Technology teaching postdoctoral fellowship from the Council on Science and Technology at Princeton University. Her research interests include information privacy and security in distributed and cyber-physical systems. 
	For her doctoral work, she received the 2007-2008 Electrical Engineering Academic Achievement Award from Rutgers University.  She is a recipient of the NSF CAREER award for 2014.
\end{IEEEbiography}

\begin{IEEEbiography}[{\includegraphics[width=1in,height=1.25in,clip,keepaspectratio]{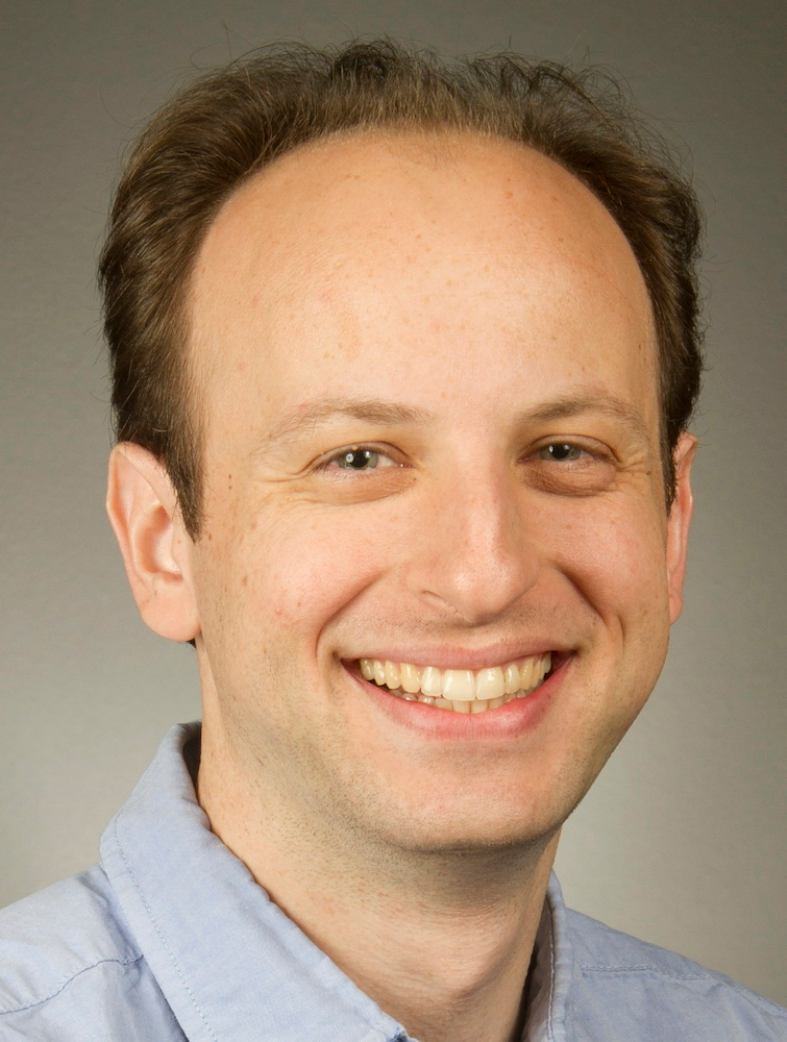}}]%
	{Oliver Kosut} (S'06--M'10) received B.S. degrees in electrical engineering and mathematics from the Massachusetts Institute of Technology, Cambridge, MA in 2004 and the Ph.D. degree in electrical and computer engineering from Cornell University, Ithaca, NY in 2010.
	
	Since 2012, he has been an Assistant Professor in the School of Electrical, Computer Engineering at Arizona State University, Tempe, AZ. Previously, he was a Postdoctoral Research Associate in the Laboratory for Information and Decision Systems at MIT from 2010 to 2012. His research interests include information theory, cyber-security, and power systems.
	
	Prof. Kosut received the NSF CAREER award in 2015.
\end{IEEEbiography}

%
%
%




\end{document}